\documentclass[11pt]{article}
\usepackage[a4paper, top=3cm, bottom=3cm, left=3cm, right=3cm]
{geometry}

\usepackage{array,xspace,multirow,hhline,tikz,colortbl,tabularx,booktabs,fixltx2e,amsmath,amsfonts,amsthm}
\usepackage{enumitem}
\usepackage{threeparttable} 
\usepackage{tablefootnote}  

\usepackage{appendix}
\usepackage{algorithmic}
\usepackage{algorithm}
\usepackage{color}

\usepackage{eqparbox}

\usepackage{tikz}
\usetikzlibrary{arrows, patterns}
\usetikzlibrary{arrows, patterns, positioning, decorations, decorations.markings}
\usepackage{pgfplots}

\usepackage[]{hyperref}
\usepackage{xcolor}

\newcommand{\E}[2]{E_{#1,#2}}
\newcommand{\ax}[1]{v_{#1}(A_{#1})}
\newcommand{\ai}[2]{v_{#1}(A_{#1}^{#2})}
\newcommand{\ao}[2]{v_{#1}(A_{#2}^{#1})}
\newcommand{\dif}[2]{\ao{#1}{#2}-\ai{#1}{#2}}
\newcommand{\ndif}[2]{\ai{#1}{#2}-\ao{#1}{#2}}
\newcommand{\Ch}{Ch}
\newcommand{\RR}{\mathsf{RoundRobin}}
\newcommand{\EC}{\mathsf{EnvyCycle}}
\newcommand{\MAXU}{\mathsf{MaxUtility}}
\newcommand{\SONE}{\mathsf{Sub1}}
\newcommand{\STWO}{\mathsf{Sub2}}

\newtheorem{theorem}{Theorem}

\newtheorem{observation}[theorem]{Observation}
\newtheorem{example}[theorem]{Example}

\newtheorem{lemma}[theorem]{Lemma}

\newtheorem{proposition}[theorem]{Proposition}

\newcommand{\cP}{\mathcal{P}}

\begin{document}

\title{On the Subsidy of Envy-Free Orientations in Graphs\footnote{The author names are ordered alphabetically}}
\author{Bo Li\thanks{Department of Computing, The Hong Kong Polytechnic University, China. comp-bo.li@polyu.edu.hk} \and
Ankang Sun\thanks{Department of Computing, The Hong Kong Polytechnic University, China. ankang.sun@polyu.edu.hk} \and  Mashbat Suzuki\thanks{School of Computer Science and Engineering, UNSW, Sydney. mashbat.suzuki@unsw.edu.au} \and
Shiji Xing\thanks{Department of Computing, The Hong Kong Polytechnic University, China. shi-ji.xing@connect.polyu.hk}}

\date{}

\maketitle

\sloppy

\begin{abstract}
    We study a fair division problem in (multi)graphs where $n$ agents (vertices) are pairwise connected by items (edges), and each agent is only interested in its incident items.  
    We consider how to allocate items to incident agents in an envy-free manner, i.e., envy-free orientations, while minimizing the overall payment, i.e., subsidy. 
    We first prove that computing an envy-free orientation with the minimum subsidy is NP-hard, even when the graph is simple and the agents have bi-valued additive valuations. 
    We then bound the worst-case subsidy. 
    We prove that for any multigraph (i.e., allowing parallel edges) and monotone valuations where the marginal value of each good is at most \$1 for each agent, \$1 each (a total subsidy of $n-1$, where $n$ is the number of agents) is sufficient.
    This is one of the few cases where linear subsidy $\Theta(n)$ is known to be necessary and sufficient to guarantee envy-freeness when agents have monotone valuations.
    When the valuations are additive (while the graph may contain parallel edges) and when the graph is simple (while the valuations may be monotone), we improve the bound to $n/2$ and $n-2$, respectively. Moreover, these two bounds are tight.
\end{abstract}

\section{Introduction}
We study envy-free allocations of indivisible goods, called items in this paper.
An allocation is envy-free (EF) if no agent strictly prefers another agent's bundle \cite{foley1967resource}. 
Unlike divisible items where EF allocations always exist \cite{brams1995envy}, an EF allocation is not guaranteed to exist when items cannot be divided.
Considering a simple example of allocating one item between two agents, both valuing the item at \$1, the agent who gets the item would be envied by the other agent.
To maintain fairness between the two agents, compensation of \$1 can be provided to the agent who receives nothing.
More generally, a compelling and pragmatic research problem is to determine the minimal compensation required to ensure EF among agents.

This problem has been widely studied in the literature within the paradigm of fair division with the minimum subsidy. Halpern and Shah \cite{DBLP:conf/sagt/HalpernS19} first proved that when agents have additive valuations and the marginal value of each good is at most 1, a total subsidy of $m(n-1)$ is sufficient to find an allocation so that no agent envies another agent. 
This result was improved to $n-1$ by Brustle et al. \cite{BDNSV20} and this bound is tight as there are instances which require at least $n-1$ subsidy in order to achieve EF.
Brustle et al. \cite{BDNSV20} also proved that for general monotone valuations an EF allocation always exists with a subsidy of $2(n-1)^2$, which was improved to $n(n-1)/2$ in \cite{DBLP:conf/aaai/KawaseMSTY24}.
It is worth noting that the tight bound for monotone valuations is still unknown \cite{DBLP:journals/jair/LiuLSW24}, where the lower bound is $n-1$.
One exception is \cite{DBLP:conf/ijcai/BarmanKNS22}, in which Barman et al. proved that a subsidy of $\Theta(n)$ is necessary and sufficient for monotone valuations with binary marginal values.

We follow the above line of research and restrict our attention to envy-free orientations in graphs -- a model that has attracted significant attention since \cite{DBLP:conf/sigecom/0001FKS23}.
Here, agents are the vertices in a graph $G=(V,E)$ and goods are the edges so that each agent is only interested in the goods incident to her, which are called {\em relevant} items of the agent.
An allocation is said to be an \emph{orientation} if every item is allocated to one of the two incident agents.
Christodoulou et al. \cite{DBLP:conf/sigecom/0001FKS23} introduced this model and proved that deciding the existence of an envy-free up to any item (EFX) orientation is NP-hard. 
An allocation is EFX if no agent envies another agent's bundle after (hypothetically) removing any item in it.
Zeng and Mehta \cite{zeng2024structureefxorientationsgraphs} characterized the structure of the graphs so that it is guaranteed that an EFX orientation exists. 
Deligkas et al. \cite{deligkas2024ef1efxorientations} strengthened the result of \cite{DBLP:conf/sigecom/0001FKS23} by proving that even with 8 agents, the problem remains NP-hard when the graph has parallel edges.
However, how to use subsidy to compensate agents so that the agents are EF has not been studied before our work.

Therefore, in this work, we address the problem of determining the amount of subsidy required to ensure envy-freeness among agents in the graph orientation problem.  
Our results span from multigraphs to simple graphs and from monotone valuations to additive valuations. 
For each setting, we consider the computation problem of finding an EF orientation with the minimum subsidy and the problem of characterizing the corresponding worst-case bound.
Our results also partially answer the open problem left in \cite{BDNSV20,DBLP:conf/ijcai/BarmanKNS22,DBLP:conf/aaai/KawaseMSTY24}, where for monotone valuations, the upper- and lower-bounds of subsidy are $O(n^2)$ and $\Omega(n)$ with the tight bound unknown \cite{DBLP:journals/jair/LiuLSW24}. 
For (multi)graph orientation problem with monotone valuations, we prove that the tight bound of subsidy is $\Theta(n)$.

\subsection{Main results}
We first consider the computation problem of computing an orientation that achieves EF among agents using the minimum subsidy. 
We then investigate the worst-case bound of subsidies that can ensure EF in various settings.
Regarding the computation problem, we prove that it is NP-hard even for bi-valued additive valuations (see Theorem \ref{thm:complexity-2value}), and is tractable when the valuations are binary additive (see Theorem \ref{thm:binary-additive}). 


\begin{table*}[ht!]
    \centering
    \begin{tabular}{|c|c|c|}
        \hline
        Setting & Upper Bound & Lower Bound  \\
        \hline
        Multi \& Monotone & $n-1$ (Theorem \ref{thm:multi:monotone:alg:n-1}) & $n-2$ (Example \ref{expl:monotone:n-2}) \\
        Multi \& Additive & $n/2$ (Theorem \ref{thm:multi:additive:alg:n/2}) & $n/2$ (Example \ref{example:additive:lower:n/2})  \\
        Simple \& Monotone & $n-2$ (Theorem \ref{thm:simple:monotone:alg:n-2}) & $n-2$ (Example \ref{expl:monotone:n-2})  \\
        \hline
    \end{tabular}
    \caption{Main Results on Bounding Subsidies}
    \label{tab:results}
\end{table*}

We then consider the problem of how much subsidy is sufficient to guarantee envy-freeness. The results are summarized in Table \ref{tab:results}. We begin with the most general case where agents have monotone valuation functions and the graph can contain parallel edges (i.e., multigraph). 
We prove that a subsidy of $n-1$ is sufficient to guarantee EF for all instances.
For the lower bound, there exists a simple graph instance where every EF orientation requires a total subsidy of at least $n-2$. 
Thus our result is tight up to an additive constant 1, leaving the exact bound as an open problem.
It is worth mentioning that this is one of the few cases where linear subsidy $\Theta(n)$ is necessary and sufficient to guarantee EF when agents have monotone valuations \cite{BDNSV20,DBLP:conf/ijcai/BarmanKNS22,DBLP:conf/aaai/KawaseMSTY24}.


For two special cases, we prove that the results can be improved.

The first improvement, which is the mostly technically involved result in our paper, is for additive valuations, while the graph can still contain parallel edges. 
Interestingly, for this case, we prove that a subsidy of $n/2$ is necessary and sufficient to ensure an EF orientation.
In comparison, without orientation restrictions, the bound of subsidy is $n-1$. 
To prove this result, we introduce {\em reserve graphs}, where each agent claims exactly one edge in which she has the highest value (called {\em reserve edge} of the agent). 
It is a simple case if no two reserve edges share the same neighborhoods, and we prove that no subsidy is required.
When {both} two agents claim reserve edges from items between them, we prove that a collective subsidy of no more than 1 is sufficient to ensure EF for all agents that these two competing agents can reach in the reserve graph.

The second improvement is for simple graphs, while the valuations are still monotone. 
Since the lower bound instance for the general case considers simple graphs, a subsidy of $n-2$ is necessary in the worst case.
We show that $n-2$ is also sufficient to ensure EF orientations for all simple graphs.
To prove this result, we show that when $n\ge 3$, there is always an envy-freeable allocation such that there are two agents who do not envy other agents before allocating any subsidy.

\subsection{Related works}
Since EF allocations may not exist, the majority of the literature is focused on its relaxations when subsidy is not allowed. Two widely adopted relaxations are envy-free up to one item (EF1) \cite{budish2011combinatorial} and envy-free up to any item (EFX) \cite{DBLP:journals/teco/CaragiannisKMPS19}, which requires no agent strictly prefers another agent's bundle excluding one or any item.
EF1 allocations always exist \cite{DBLP:conf/sigecom/LiptonMMS04} but the existence of EFX allocations is still unknown except for several special cases, e.g., \cite{DBLP:journals/siamdm/PlautR20,DBLP:journals/jacm/ChaudhuryGM24}. 
We refer to the recent surveys \cite{DBLP:journals/ai/AmanatidisABFLMVW23,DBLP:journals/jair/LiuLSW24} for a comprehensive overview of fair allocation.
Below we restrict our focus on the graph orientation and fair division with subsidy.

\paragraph{Graph Orientation}
Christodoulou et al. \cite{DBLP:conf/sigecom/0001FKS23} first introduced the graph model where vertices are agents and edges are items such that agents are only interested in their incident items.
When agents have additive valuations, they proved that an EFX allocations when the items can be arbitrarily allocated always exists but an EFX orientation where items must be allocated to their incident agents may not exist.
Zeng and Mehta \cite{zeng2024structureefxorientationsgraphs} identified a family of graphs which admit EFX orientations exist. 
Kaviani et al. \cite{DBLP:journals/corr/abs-2407-05139} and Deligkas et al. \cite{deligkas2024ef1efxorientations} 
repectively generalized the graph orientation problem.
In the model of Kaviani et al. \cite{DBLP:journals/corr/abs-2407-05139}, agents have $(p,q)$-bounded valuations, where each item has a non-zero marginal value for at most $p$ agents, and any two agents share at most $q$ common interested items. 
They proved that for $(\infty,1)$-bounded valuations, EF2X orientations always exist.
In the model of Deligkas et al. \cite{deligkas2024ef1efxorientations}, every agent has a predetermined set of items and only receives items from this set. They proved that for monotone valuations, an EF1 allocation always exists. 
Zhou et al. \cite{zhoucomplete} extended the goods model in \cite{DBLP:conf/sigecom/0001FKS23} to mixed setting when the marginal value of each item can be positive or negative.



\paragraph{Fair Division with Subsidy}
Caragiannis and Ioannidis \cite{DBLP:conf/wine/CaragiannisI21} first studied the computational complexity of using the minimum amount of subsidies to ensure envy-freeness.
Halpern and Shah \cite{DBLP:conf/sagt/HalpernS19} bounded the worst-case amount of subsidies for additive valuations, and Brustle et al. \cite{BDNSV20} later improved the result in \cite{DBLP:conf/sagt/HalpernS19}.
There is a line of recent works focusing on the variants and generalizations of the standard setting.
For example, Goko et al. \cite{DBLP:journals/geb/GokoIKMSTYY24} considered the strategic behaviors when subsidy is used, and Aziz et al. \cite{DBLP:journals/corr/abs-2408-08711} considered weighted envy-freeness when the agents have different entitlements.
Brustle et al. \cite{BDNSV20}, Barman et al. \cite{DBLP:conf/ijcai/BarmanKNS22} and Kawase et al. \cite{DBLP:conf/aaai/KawaseMSTY24} extended the valuations from additive to general monotone, which is also the focus of the current paper.
Choo et al. \cite{DBLP:journals/orl/ChooLSTZ24} and Dai et al. \cite{DBLP:journals/corr/abs-2408-12523} respectively studied the envy-free house allocation with subsidies and its weighted variant. 
Wu et al. \cite{DBLP:conf/wine/WuZZ23, DBLP:journals/corr/abs-2404-07707} extended the subsidy model to the fair allocation of indivisible chores.

\section{Preliminaries}
We first introduce the necessary notations and solution concepts in standard fair division problems.
Denote by $[k] = \{ 1,\ldots, k \}$ for all $k \in \mathbb{N}^+$. There is a set $M = \{ e_1,\ldots, e_m \}$ of $m$ items to be allocated to a set $N=\{1,\ldots, n \}$ of $n$ agents.
An allocation $\mathbf{A} = (A_1,\ldots, A_n)$ is an $n$-partition of $M$ where $A_i$ refers to the bundle allocated to agent $i$.
Each agent $i$ has a monotone valuation function $v_i(\cdot): 2^{M}\rightarrow \mathbb{R}_{+}$ such that for any $S\subseteq T$, $v_i(S) \leq v_i(T)$. 
Following the convention of the literature on fair division with subsidies, we scale the valuations such that the maximum marginal value of any item is 1. 
That is $\max_{e\in M, S\subseteq M}v_i(S\cup\{e\}) - v_i(S) = 1$.
The valuation $v_i(\cdot)$ is \emph{additive} if for any $S\subseteq M$, $v_i(S)=\sum_{e\in S}v_i(\{e\})$. 
We say item $e$ is \emph{irrelevant} for agent $i$ if 
$v_i(S\cup\{e\}) = v_i(S)$ for all $S\subseteq M$.
In other words, the marginal value of assigning an agent her irrelevant item is always 0.
Item $e$ is \emph{relevant} for agent $i$ if it is not irrelevant.

With no subsidies, an allocation is \emph{envy-free} (EF) if no agent would prefer the bundles of others to her own bundle, i.e., for any $i,j\in N$, $v_i(A_i) \geq v_i(A_j)$. 
Let $\mathbf{p}=(p_1,\ldots,p_n)$ be the \emph{subsidy} or payment vector where each agent $i$ receives payment $p_i\geq 0$. An allocation with payment $\{ \mathbf{A}, \mathbf{p} \}$ is envy-free if for any $i,j\in N$,
$$
v_i(A_i) + p_i \geq v_i(A_j) + p_j.
$$
An allocation $\mathbf{A}$ is said to be \emph{envy-freeable} (EFable) if there is a payment vector $\mathbf{p}=(p_1,\ldots,p_n)$ such that $\{\mathbf{A}, \mathbf{p}\}$ is envy-free. 
It is known that not every allocation is envy-freeable. 
Halpern and Shah \cite{DBLP:conf/sagt/HalpernS19} characterized all envy-freeable allocations via \emph{envy graph}.
The envy graph $D_{\mathbf{A}}$ for allocation $\mathbf{A}$ is a complete weighted directed graph with vertices $N$. For any pair of agents $i,j\in N$, the weight of arc $(i,j)$ in $D_{\mathbf{A}}$ is $w(i,j) = v_i(A_j) - v_i(A_i)$, i.e., the envy of agent $i$ on agent $j$'s bundle.

\begin{theorem}[Halpern and Shah \cite{DBLP:conf/sagt/HalpernS19}]\label{thm:condition-envyfreeable}
    Given any allocation $\mathbf{A}$, the following statements are equivalent:
    \begin{enumerate}[label=(\alph*)]
        \item $\mathbf{A}$ is envy-freeable;
        \item $\mathbf{A}$ maximizes the total value of agents among all reassignments of its bundles to agents, i.e., for every permutation $\pi: N\rightarrow N$, $\sum_{i\in N}v_i(A_i) \geq \sum_{i\in N}v_i(A_{\pi(i)})$;
        \item $D_{\mathbf{A}}$ contains no directed cycle with positive weight.
    \end{enumerate}
\end{theorem}

\paragraph{Graph Orientation}

In this work, we assume that agents are vertices in a (multi)graph and items are edges so that the relevant items to an agent are exactly the edges incident to her. 
Formally, given a graph $G=(V,E)$, we use $E_{i,j}$ to denote the set of edges between $i,j\in V$, where $E_{i,j}=E_{j,i}$ and $E_{i,j}=\emptyset$ if $i,j$ are non-adjacent.
Let $E_i=\bigcup_{j\neq i}E_{i,j}$ be the set of edges incident to $i \in V$, i.e, $i$'s relevant items.
When the graph is simple, $E_{i,j}$ contains at most one edge and we use $e=(i,j)$ to denote this edge if $E_{i,j}\neq \emptyset$. 
Then $E_i=\{(i,j)\in E\}$ for a simple graph.
In summary, we have $V=N$ and $M=E$ such that $v_i(X) = v_i(X\cap E_i)$ for all $X\subseteq E$.
We focus on a restricted type of allocations.
An \emph{orientation} with respect to $G$ is an allocation $\mathbf{A} = (A_1,\ldots, A_n)$ such that $A_i \subseteq E_i$.
If not explicitly stated otherwise, all allocations in this paper are assumed to be orientations.

Given orientation $\mathbf{A}=(A_1,\ldots,A_n)$ of $G$, for any $i,j\in N$, let $A^j_i = A_i\cap E_{i,j}$ for simplicity.
The orientation is \emph{locally envy-freeable} with respect to agents $i,j$
if the partial allocation of $E_{i,j}$ is envy-freeable, i.e., $v_i(A^j_i)+v_j(A^i_j) \geq v_i(A^i_j)+v_j(A^j_i)$.
The orientation $\mathbf{A}$ is locally envy-freeable if for every pair of agents $i,j$, the partial allocation of $E_{i,j}$ is envy-freeable.

We without loss of generality assume throughout this paper that the given (multi)graph is connected. If the graph consists of multiple components, we can add new edges until there exists exactly one component. Denote the resulting graph as $G'$.
Let the agents incident to each new-added edge have value 0 for the edge. Then an EF orientation with payment $\{\mathbf{A}, \mathbf{p}\}$ in $G$ can be converted to $\{\mathbf{A}',\mathbf{p}\}$ an EF orientation with payment in $G'$; this can be achieved by arbitrarily distributing the new-added edges to incident agents.
The reverse direction is also true. As only the orientation varies, we can claim that there exists an EF orientation of $G$ with total subsidy at most a number $P$ if and only if there exists an EF orientation of $G'$ with total subsidy at most $P$.





\subsection{Algorithmic Components}

In this section, we recall some simple algorithms in resource allocation, which will be used as subroutines in our algorithms. 
Although these algorithms work for an arbitrary number of agents, it is sufficient for us to focus on the case with two agents. 

\paragraph{Round-Robin ($\RR(a,b,S)$)} The algorithm takes a pair of ordered agents $a$ and $b$ and a set of items $S$. Agents $a$ and $b$ take turns selecting their favorite remaining item from $S$ until all items are allocated. 
When the two agents have additive valuations, $\RR$ returns an EF1 allocation. 

\paragraph{Envy-Cycle Elimination ($\EC(a,b,S)$)}
The algorithm takes two agents $a$ and $b$ and a set of items $S$. 
In each round, $\EC$ selects the agent that is not envied by the other one to receive one remaining item in $S$ until all items are allocated. 
If there is mutual lack of envy between both agents, arbitrarily select one of them. 
If both agents envy the other one, $\EC$ first swaps their current allocations before selection. 
$\EC$ returns an EF1 allocation for monotone valuations.

\paragraph{Maximum Total Utilities ($\MAXU(a,b,S)$)}
For each item in $S$, it is allocated to the agent who has a higher value on it, and ties are broken arbitrarily.
The algorithm always maximizes the total value of the two agents but does not necessarily have fairness guarantees.

\section{Complexity of Computing Minimum Subsidy}\label{sec:complexity}

In this section, we present our computational results on the problem of computing an EF orientation with the minimum total subsidy when the valuations are additive. We demonstrate that, perhaps surprisingly, the problem is NP-hard if the edge values are in $\{1,2\}$, but is polynomial-time solvable if edge values are  in $\{0,1\}$.

\subsection{NP-hardness for Edge Values in $\{1,2\}$}
We start with the hardness result. 
For the ease of presentation, we consider the valuations with $v_i(e)\in \{1,2\}$ for all $i\in N$ and $e\in E_i$.
We can simply divide all the values by 2 in order to satisfy the scaled assumption of maximum marginal value being 1. 

\begin{theorem}\label{thm:complexity-2value}
    Computing an orientation with the minimum subsidy is the NP-hard, even when the graph is simple and the valuations are additive and $v_i(e)\in \{1,2\}$ for all $i\in N$ and $e \in E_i$.
\end{theorem}

\begin{proof} 
    It suffices to prove the NP-completeness of checking whether an EF orientation exists or not. We reduce from the NP-complete problem 2P2N-3SAT \cite{DBLP:journals/dam/BermanKS07,DBLP:conf/rta/Yoshinaka05}. A 2P2N-3SAT instance contains a Boolean formula in conjunctive normal form consisting of $n$ Boolean variables $\{x_i\}_{i\in[n]}$ and $m$ clauses $\{C_j\}_{j\in[m]}$.
    Each variable appears exactly twice as a positive literal and exactly twice as a negative literal in the formula. 
    Each clause contains three distinct literals.

  
    Given a 2P2N-3SAT instance $(\{x_j\}_{j\in [n]}, \{C_j\}_{j\in [m]})$, we construct a graph orientation instance without parallel edges as follows:
    \begin{itemize}
        \item For each variable $x_i$, create two vertices $i,\bar{i}$ and one edge $(i,\bar{i})$ with a value of 2 to both $i$ and $\bar{i}$.
        \item For each clause $C_j$, create a vertex $V_j$. If $C_j$ contains a positive literal $x_i$, then create an edge $(V_j, i)$ with value 1 for both incident vertices. If $C_j$ contains a negative literal $\bar{x}_i$, then create an edge $(V_j, \bar{i})$ with value 1 for both incident vertices.
        \item Create a dummy vertex $D_j$ for every clause $C_j$, and an edge $(D_j, V_j)$ with a value of 1 for both incident vertices. 
    \end{itemize}
    In summary, we create a graph orientation instance with $2m+2n$ vertices and $5n+m$ edges. Only the edges of form $(i,\bar{i})$'s are valued at 2 by their incident vertices, and all other edges have a value of 1 by their incident vertices. Observe that the degree of each dummy vertex is 1, and thus, in an EF orientation, every vertex $D_j$ must receive edge $(D_j,V_j)$. 
We now show that a 2P2N-3SAT instance is satisfiable if and only if there exists an EF orientation. 

Suppose the 2P2N-3SAT instance is satisfiable. For each variable $x_i$ that is set to \texttt{True} in the satisfiable assignment, we allocate $(i,\bar{i})$ to $i$. Then, allocate the other two edges incident to $i$ to the corresponding clause vertices.
For $\bar{i}$, there are also two edges linking two clause vertices to it, and then allocate these two edges to $\bar{i}$. Similarly, if $x_j$ is set to \texttt{False}, allocate  $(j,\bar{j})$ to $\bar{j}$. Additionally, allocate the other two edges incident to $\bar{j}$ to their corresponding clauses. For $j$, allocate the edges connecting to the clause vertices to $j$. Finally, for every edge $(D_j,V_j)$,  allocated it to $D_j$. We have allocated all edges to their incident vertices, it remains to show that this allocation is EF. 
Each vertex $D_j$ is EF since she receives her unique-valued item. Each $V_j$ receives a value of at least 1 as there exists at least one literal contained in $C_j$ that evaluates to \texttt{True} in the satisfiable assignment. Thus, $V_j$ is also EF.
    For variable $x_i$ that is set to \texttt{True}, the vertex $i$  receives an item with a value of 2, and $\bar{i}$ receives 
    two incident edges linking to clause vertices and has a value of 2. A similar argument shows that when $x_j$  is set to \texttt{False}, both  $j$ and $\bar{j}$ receive value 2. Therefore, the resulting orientation is indeed EF.

    For the reverse direction, suppose that there exists an EF orientation. Construct an assignment of $\{x_i\}$ as follows: if $(i,\bar{i})$ is allocated to $i$, then set $x_i$ to \texttt{True}, and otherwise, set $x_i$ to \texttt{False}.
    The constructed assignment is valid, which means that exactly one of $x_i,\bar{x}_i$ is set to \texttt{True}. We now show that this assignment is satisfiable. Assume, for the sake of contradiction, that there exists a clause $C_j$ that evaluates to \texttt{False}. 
    In this case, the vertices corresponding to the three literals in $C_j$ would not receive the item with value 2. To satisfy EF, each of these vertices must receive the edge linking her to $V_j$.
    Consequently, $V_j$ needs to receive $(D_j,V_j)$ to avoid envying the variable vertices adjacent to her. However, this allocation would make $D_j$ envy $V_j$, which is a contradiction.
    Therefore, the constructed assignment is a satisfiable assignment.
\end{proof}
We remark that the above theorem indicates that it is NP-hard to differentiate between zero subsidy and a positive subsidy. Therefore, there is no polynomial time algorithm with a bounded approximation guarantee for the problem of computing the orientation with the minimum subsidy. 
Given this, we study the subsidy bounds that always guarantees an EF orientation in Sections \ref{sec:Monotone}, \ref{sec::n/2}, and \ref{sec:simple:monotone}.

\subsection{Efficient Algorithm for Edge Values in $\{0,1\}$}
We next prove that when the valuation functions are additive and binary (i.e., the value of every edge is in $\{0,1\}$), computing the minimum subsidy orientation can be solved in polynomial time.
{Starting from $G$, if an edge is valued at 0 by both incident vertices, we remove that edge. Note that the removal of those edges does not affect the required subsidy. 
With a slight abuse of notation, for the remainder of this section, let $G$ denote the graph obtained after removing the above-mentioned edges.} 
If an edge is valued at 1 by both incident vertices, it is called \emph{critical}.
If it is valued at 1 by exactly one of its incident vertices, it is called \emph{non-critical}.

 We now introduce four properties that a component $C=(V_C,E_C)$ of $G$ may satisfy: 
\begin{enumerate}[label=(P\arabic*),ref=P\arabic*]
    \item $C$ contains a non-critical edge.
    \item  $C$ contains a simple cycle\footnote{A simple cycle with length $k$ in a multigraph $G$ is a simple path $P = v_1,e_1,v_2,e_2,\ldots,v_k,e_k,v_1$ such that vertices $v_1,\ldots,v_k$ are distinct, and $e_i \in E_{i,i+1}$ for $1\le i<k$ and $e_k\in E_{k,1}$.} of length  at least 3 where every edge is critical.
    \item $C$ contains a pair of {adjacent} vertices $i$ and $j$ such that the number of critical edges in $E_{i,j}$ is even.
    \item $C$ contains vertices $i, j, k$ such that there are at least two critical edges in $E_{i,j}$, and at least two critical edges in $E_{i,k}$.
\end{enumerate}


\begin{algorithm}[t]
	\caption{Computing an EF Orientation}
	\label{alg:subroutine}
 \renewcommand{\algorithmicensure}{\textbf{Output:}}
	\begin{algorithmic}[1]
		\REQUIRE {A component $C=(V_C,E_C)$ satisfying any of (P1)--(P4).}
		\ENSURE An EF orientation of $C$.
        \STATE Initialize the label of each vertex $v\in V_C$ to $\text{Label}_v=0$.
        \STATE If property (P1) is satisfied, allocate each {non-critical} edge to the agent who values it at 1, and set that agent's label to 1. Remove all allocated edges and proceed to Line \ref{alg:subroutine:step-while}.
        \STATE If property (P2) is satisfied, arbitrarily choose a direction for the cycle and allocate each agent one edge along the chosen direction. Additionally, set the label of every agent in the cycle to 1.
        Remove all allocated edges and proceed to to Line \ref{alg:subroutine:step-while}.
        \STATE If property (P3) is satisfied, allocate half of the critical edges to each of the agents $i$ and $j$. Set the labels of $i$ and $j$ to the number of edges they receive. Remove all allocated edges and proceed to Line \ref{alg:subroutine:step-while}.
        \STATE If property (P4) is satisfied, allocate $\lfloor \frac{\ell_{ij}}{2} \rfloor$ and $\lfloor \frac{\ell_{ik}}{2} \rfloor$ critical edges to agent $i$, where $\ell_{ij}$ is the number of {critical} edges between $i,j$ in $C$. Set the label of agent $i$ to the total number of edges they receive. Remove all allocated edges and proceed to Line \ref{alg:subroutine:step-while}.
        \WHILE{$V_C\neq \emptyset$}\label{alg:subroutine:step-while}
        \STATE\label{Line7} Pick an arbitrary agent $i\in V_C$ with $\text{Label}_i \geq 1$. For each neighbor $j$ of $i$, suppose $\ell_{ij}$ the number of {remaining} edges\footnotemark{}
        between $i$ and $j$ in $C$. Let agent $i$ receive $\max\{ \lceil \frac{\ell_{ij}}{2} \rceil - \text{Label}_i, 0 \}$ of the edges connecting $i$ and $j$. Remove vertex $i$ from $C$ (i.e., $V_C \leftarrow V_C\setminus\{i\}$ ), and remove all allocated edges.
        \STATE\label{Line8} If there are edges remaining  that only connects to one vertex, allocate those edges to the incident vertices. Additionally, set the label of each these vertices  to the  number of edges allocated to them. 
        \ENDWHILE
	\end{algorithmic}
\end{algorithm}

\footnotetext{Note that when the while loop is executed, there is no non-critical edge.}

We first prove that a component admits an EF orientation if and only if it satisfies at least one of properties (P1)--(P4).

\begin{lemma}\label{lem:subroutine-ak-ec-1}
    For a component $C$ of $G$, there exists an \textnormal{EF} orientation of $C$ if and only if $C$ satisfies at least one of properties (P1)--(P4).
\end{lemma}
\begin{proof}
    For the ``if'' direction, we prove that Algorithm~\ref{alg:subroutine}  terminates and returns an EF orientation. 
    Since the input $C$ satisfies at least one of the properties (P1)--(P4), first time when  the while loop in Line \ref{alg:subroutine:step-while} is executed, there exists at least one agent with a positive Label.
    In the while-loop, during each iteration,  we remove one vertex and moreover set the Label of at least one other remaining vertex to a positive integer. Therefore, the while-loop terminates with $V_C=\emptyset$ and all edges are allocated.

    Next, we prove the resulting orientation is EF. Observe that when the while loop is executed, we can assume that every edge is critical; otherwise, this would imply that property (P1) is satisfied, and thus every 0-1 edge could have  been removed.  Note that it suffices to check the EF condition only among neighboring vertices $i$ and $j$, since if there are no edges between two different vertices, there cannot be any envy between them. 

    Let  $i$ and $j$ be arbitrary adjacent vertices. One of these two vertices must have been selected earlier in the while loop (without loss of generality, let this vertex be $i$). When $i$ is selected by the while loop, her value is equal to $\text{Label}_i$. Thus, when the vertex $i$ is removed in Line~\ref{Line7}, her value is at least $\lceil \frac{\ell_{ij}}{2} \rceil$, ensuring that $i$ does not envy $j$. Since the remaining edges $\ell_{ij}-\max\{ \lceil \frac{\ell_{ij}}{2} \rceil - \text{Label}_i, 0 \} \geq \lceil \frac{\ell_{ij}}{2}\rceil$ are all allocated to agent $j$ by Line~\ref{Line8}, we also have that $j$ does not envy $i$. Thus, any two neighboring vertices are EF toward each other, implying that the outputted orientation is EF, as needed.

    We now consider the ``only if'' direction and prove that if $C$ satisfies none of (P1)-(P4), then $C$ does not have an EF orientation.
    For the sake of contradiction, assume that $C=(V_C,E_C)$ does have an EF orientation. As $C$ does not satisfy property (P1), it follows that $C$ does not contain non-critical edge; thus, every edge is critical. Denote $\ell_{ij}$ as the number of edges between any two adjacent vertices $i,j \in V_C$. 
 
    Consider a directed graph $\bar{C}$, induced by the EF orientation, with vertex set $V_C$ and edge set defined as 
    \[
    E_{\bar{C}}:= \{ (i,j) \ | \ (i,j)\in E_C,\ i \text{ receives } \lceil\frac{\ell_{ij}}{2}\rceil \text{ edges} \}.
    \]
    We see that in any EF orientation, each vertex $i$ in $\bar{C}$ must have an out-degree at least one. This holds because in $C$: (i) the number of edges between any two vertices $i,j \in V_C$ is odd, and (ii) for each $i\in V_C$ with at least two neighbors, $i$ has at most one neighbor with whom it shares more than one edge. These facts imply that if $i$ has out degree zero in $\bar{C}$, it receives a value of at most $\max_{j\in V_C}\lfloor\frac{\ell_{ij}}{2}\rfloor$, implying that $i$ envies one of its neighbors. Thus, indeed the out-degree of each vertex in $\bar{C}$ must be at least one.

    A directed graph where each vertex has out-degree at least one must contain a cycle. Hence, we see that $\bar{C}$ must contain a cycle. Furthermore, since if $(i,j)\in E_{\bar{C}}$ then $(j,i)\not\in E_{\bar{C}}$, we see that the cycle must be of length at least 3. However, a cycle of length at least 3 in $\bar{C}$ implies the existence of a cycle of the same length in $C$, which means that $C$ satisfies property (P3), contradicting the assumption that $C$ does not satisfy any of the properties (P1)--(P4). Hence, by contradiction, there is no EF orientation of $C$.
\end{proof}

We now present the main result of this section.

\begin{theorem}\label{thm:binary-additive}
  {When agents have binary additive valuations}, the minimum subsidy required to ensure EF orientation is equal to the number of connected components that do not satisfy any of the properties (P1)–(P4). Furthermore, the EF orientation that achieves this  minimum subsidy can be computed in polynomial time.
\end{theorem}
\begin{proof}

Suppose $G$ consists of connected components $C_1,\ldots,C_r$, $H_1,\ldots,H_t$, where each $C_i$ satisfies one of the properties (P1)--(P4), while each $H_j$ satisfies none of them. We show that subsidy of $t$, corresponding to one dollar per each component not satisfying (P1)--(P4), is both necessary and sufficient for $G$ to have an EF orientation. 

By Lemma~\ref{lem:subroutine-ak-ec-1}, each component $C$ has an EF orientation, and each vertex $i\in V_C$ receives a value of at least 1. Thus for each $i\in V_C$, if some agent $j\in V\setminus V_C$ receives a subsidy of at most 1, then $i$ is EF towards $j$.
Then to establish the subsidy bound in the theorem statement, it suffices to show that for each component $H$ not satisfying any of (P1)--(P4), subsidy of 1 is both necessary and sufficient for $H$ to have an EF orientation; this guarantees that no envy between vertices from different components.

According to Lemma~\ref{lem:subroutine-ak-ec-1}, for the component $H=(V_H,E_H)$, there exists no EF orientation, indicating that a positive subsidy is necessary. 
We then demonstrate that a subsidy of 1 suffices. Select an arbitrary vertex $i\in V_H$, give $i$ a subsidy of 1, and set $\text{Label}_i=1$. Then, proceed  directly to the while loop in  Algorithm \ref{alg:subroutine}.  The arguments provided in Lemma~\ref{lem:subroutine-ak-ec-1} imply that the while loop terminates and that the resulting orientation is EF. As only one agent received subsidy 1, we conclude that $H$ has an EF orientation with subsidy 1.

Finally, in order to compute the EF orientation of $G$, we essentially run Algorithm~\ref{alg:subroutine} on each of the components of $G$, with a minor adjustment for those components that do not satisfy any of the properties (P1)--(P4). Since the running time of Algorithm~\ref{alg:subroutine} is $O(mn)$, and the number of components of $G$ is at most $\frac{n}{2}$, we see that the EF orientation of $G$ with the minimum subsidy can be computed in time $O(mn^2)$.
\end{proof}

\section{Subsidy Bounds for Monotone Valuations and Multigraphs}\label{sec:Monotone}

In this section, we focus on the 
general case where the graph contains parallel edges and the valuations are monotone. {Suppose that valuations are given as oracle.}
Our main result is a polynomial-time algorithm that outputs an EF orientation using at most $n-1$ subsidy. 
As we will see in Section \ref{sec:simple:monotone} (Example \ref{expl:monotone:n-2}), there are instances where every EF orientation requires a subsidy of at least $n-2$ (even when there are no parallel edges), and thus our result in this section is tight up to an additive constant of 1.





Before presenting our algorithm, we first observe that  achieving local envy-freeability between every pair of adjacent agents 
does not imply envy-freeability for the orientation as a whole. We consider the following example.

\begin{example}
    Consider a complete graph with 3 agents $\{1,2,3\}$. We first construct $\E{1}{2}$, items between agents $1$ and $2$, and the corresponding valuations. 
    The set of items $E_{1,2}$ is partitioned into two bundles $B_1,B_2$ with agents' valuations: 
    $v_1(B_1) = 1, v_1(B_2) = \frac{2}{3}, v_2(B_1) = \frac{2}{3}, v_2(B_2) = 0$. 
    A local envy-freeable partial allocation is to assign
    $B_1$ to agent 2 and $B_2$ to agent 1.
    For the set $\E{2}{3}$ of items, we also split it into two bundles $B_3,B_4$, and the valuations of agents are $v_2(B_3) = 1, v_2(B_4) = \frac{2}{3}, v_3(B_3) = \frac{2}{3}, v_3(B_4) = 0$. 
    Then allocate $B_4$ to agent 2  and $B_3$ to agent 3  resulting in a locally envy-freeable partial allocation. 
    Similarly, the set $\E{1}{3}$ of items is partitioned into $B_5$ and $B_6$, with agents' valuations being $v_3(B_5) = 1, v_3(B_6) = \frac{2}{3}, v_1(B_5) = \frac{2}{3}, v_1(B_6) = 0$. Similarly, the partial allocation where agent $1$ receives $B_5$ and agent $3$ receives $B_6$ is locally envy-freeable.
    Now we get an orientation of $\mathbf{A}$ where $A_1=B_2\cup B_5, A_2=B_1\cup B_4, A_3=B_3\cup B_6$.
    As for agents' valuations of their bundles, we define $v_1(A_1)=v_2(A_2)=v_3(A_3)=\frac{2}{3}$, which does not violates the monotonicity of $v_i$'s.
    However, since in the envy graph $D_{\mathbf{A}}$, the weight of cycle $1\rightarrow 2\rightarrow 3\rightarrow 1$ is 1, the orientation $\mathbf{A}$ is not EFable by Theorem \ref{thm:condition-envyfreeable}.    
\end{example}

We circumvent such technical difficulties and show that, for multigraphs, a total subsidy of $n-1$ is sufficient. 
We now establish a structural property that will be useful later.


\begin{algorithm}[ht!]
        \caption{Multigraph Orientation for Monotone Valuations}
	\label{alg:subsidy-monotonoe}
	\begin{algorithmic}[1]
		\REQUIRE An instance $G=(N, E)$ and the valuation functions $\{ v_i \}'s$.
		\ENSURE Orientation $\mathbf{A} = (A_1,\ldots,A_n)$.
        \STATE For each pair $i,j$ of adjacent agents, apply the  envy-cycle elimination algorithm of Lipton et al. \cite{DBLP:conf/sigecom/LiptonMMS04} to the items $E_{i,j}$ and agents $i$ and $j$. 
        Let $T^i_{i,j}$ and $T^j_{i,j}$ be the resulting allocations, representing the bundles temporarily allocated to agents $i$ and $j$ respectively.
        \STATE For every agent $i$, define a threshold value
        $$
        b_i=\max\limits_{j\in D(i)}\min \{ v_i(T^i_{i,j}), v_i(T^j_{i,j}) \},
        $$
        where $D(i)$ is the set of neighbors of agent $i$ in graph $G$.
        \STATE For each pair of adjacent agents $i,j$, if they do not envy each other with respect to the temporary  allocation between them, i.e., $v_i(T^i_{i,j}) \geq v_i(T^j_{i,j})$ and $v_j(T^j_{i,j}) \geq v_j(T^i_{i,j})$,  make the temporary allocation permanent by setting   $P^i_{i,j}=T^i_{i,j}$ and $P^j_{i,j}=T^j_{i,j}$.
        \STATE\label{alg:line:temp} For each pair of adjacent agents $i,j$, if there exists an envious agent with respect to the temporary allocation between them (note that there are no envy-cycles after applying the algorithm of Lipton et al. \cite{DBLP:conf/sigecom/LiptonMMS04}),
        we assume both agents $i$ and $j$ prefer $T^i_{i,j}$ to $T^j_{i,j}$.  If $v_i(T^i_{i,j})-b_i \geq v_j(T^i_{i,j})-b_j$, we   permanently allocate $T^i_{i,j}$ to agent $i$ and the other bundle to agent $j$; otherwise, we permanently allocate $T^i_{i,j}$ to agent $j$ and the other bundle to agent $i$.
        \STATE For every agent $i,$ let $P^i_{i,j} \subseteq E_{i,j}$ be the  bundle allocated permanently to agent $i$, and the whole bundle received by agent $i$ is $A_i=\bigcup\limits_{j\in D(i)} P^i_{i,j}$.
	\end{algorithmic}
\end{algorithm}

\begin{lemma}\label{lem:envy-cycle}
The allocation $\mathbf{A}$ returned by Algorithm~\ref{alg:subsidy-monotonoe} does not contain any envy-cycle. 
\end{lemma}
\begin{proof}
Consider a directed graph $\bar{G}$ where every agent $i$ corresponds to a vertex and arc $(i,j)$ from $i$ to $j$ exists if and only if agent $i$ envies $j$ in allocation $\mathbf{A}$. We now show there is no cycle in $\bar{G}$. 
Suppose, for the sake of a contradiction, there exists a cycle $\mathcal{C}=\{1,2,\ldots, h,1\}$ in $\bar{G}$. For notational convenience,  let agent $h+1$ be $1$. For any $j \leq h$, agent $j$ envies $j+1$, and thus $j$ and $j+1$ are adjacent vertices in $G$ and $E_{j,j+1}\neq\emptyset$.
Moreover, by Line~\ref{alg:line:temp}, when determining the partial allocation of $E_{j,j+1}$, both agents $j$ and $j+1$ prefer the same bundle between $T^j_{j,j+1}$ and $T^{j+1}_{j,j+1}$. Suppose that both prefer $T^j_{j,j+1}$ than the other bundle.
Since agent $j$ envies $j+1$ in allocation $A$, bundle $T^j_{j,j+1}$ is permanently allocated to agent $j+1$; note in this case, $T^j_{j,j+1}=P^{j+1}_{j,j+1}$.
    By Line~\ref{alg:line:temp}, for any $j\in [h]$, we have 
    $$
    v_j(P^{j+1}_{j,j+1}) - b_j \leq v_{j+1}(P^{j+1}_{j,j+1}) - b_{j+1}.
    $$
    Summing up $j\in [h]$ and eliminating $\sum_{j\in [h]} b_j$ from both sides,
    \begin{equation}\label{eq:ak-ad1}
    \sum_{j=1}^{h} v_j(P^{j+1}_{j,j+1}) \leq \sum_{j=1}^{h} v_{j+1}(P^{j+1}_{j,j+1}), 
    \end{equation}
    where $P^{h+1}_{h,h+1}$ is identical to $P^1_{h,1} \subseteq E_{h,1}$, the bundle permanently allocated to agent 1.
    Moreover, since $\mathcal{C}$ is a cycle in $\bar{G}$, for any $j\in [h]$, we have
    \begin{equation}\label{eq:ak-ad2}
    v_j(A_j)< v_j(A_{j+1}) = v_j(P^{j+1}_{j,j+1}),
    \end{equation}
    where the equality transition is due to the property of orientation.
    Combining inequalities (\ref{eq:ak-ad1}) and (\ref{eq:ak-ad2}), we have
    \begin{align*}
        \sum_{j=1}^h v_j(A_j) &\geq v_1(P^1_{h,1}) + \sum_{j=1}^{h-1} v_{j+1}(P^{j+1}_{j,j+1}) \\
        &\geq v_h(P^1_{h,1}) + \sum_{j=1}^{h-1}v_j(P^{j+1}_{j,j+1}) > \sum_{j=1}^h v_j(A_j),
    \end{align*}
    where the first inequality holds  because for any $j\in [h]$, $P^{j+1}_{j,j+1}\subseteq A_{j+1}$ and valuation functions are monotone; the second  transition is due to inequality (\ref{eq:ak-ad1}); the third inequality transition is due to inequality (\ref{eq:ak-ad2}).
    Thus, by contradiction, there is no cycle $\mathcal{C}$ in $\bar{G}$.
\end{proof}

\begin{lemma}\label{lem:mono-multi-efable}
    Algorithm~\ref{alg:subsidy-monotonoe} returns an envy-freeable allocation $\mathbf{A}$.
\end{lemma}

We now present the upper bound on the subsidy for each agent. The result of
\cite{DBLP:conf/sagt/HalpernS19} states that for any envy-freeable allocation $\mathbf{A}$, a valid subsidy payments which ensures  $\{ \mathbf{A}, \mathbf{p} \}$ envy-free is to set $p_i=\ell^{\max}_{D_{\mathbf{A}}}(i)$, where $\ell^{\max}_{D_{\mathbf{A}}}(i)$ is the weight of the maximum weight directed path starting from vertex $i$ in the envy graph $D_{\mathbf{A}}$. For the allocation returned by Algorithm~\ref{alg:subsidy-monotonoe}, we also set the payments as $p_i=\ell^{\max}_{D_{\mathbf{A}}}(i)$. We now show that, for the returned allocation, weight of the maximum weight directed path is at most 1, which  implies $p_i=\ell^{\max}_{D_{\mathbf{A}}}(i)\leq 1 ,\forall i\in N$.

\begin{lemma}\label{lem:weight-length-mono}
    Let $\cP$ be the maximum weight directed path in the envy graph $D_{\mathbf{A}}$, where $\mathbf{A}$ is the allocation returned by Algorithm~\ref{alg:subsidy-monotonoe}. The weight of $\cP$ is at most 1.
\end{lemma}

Combining Lemmas \ref{lem:mono-multi-efable} and \ref{lem:weight-length-mono} gives our main result.
\begin{theorem}
\label{thm:multi:monotone:alg:n-1}
  For multigraphs with monotone agent valuations, there always exists an EF orientation with a total subsidy of at most $n-1$, and it can be computed in polynomial time using value queries.
\end{theorem}

\section{Improved Bounds for Additive Valuations}\label{sec::n/2}

In this section, we improve Theorem \ref{thm:multi:monotone:alg:n-1} for additive valuations, where the graph can still contain parallel edges. 
As we will see, breaking $n-\Theta(1)$ is far more technically involved.
We first present an instance for which we can conclude that a subsidy of $n/2$ is necessary to achieve an EF orientation when valuations are additive.

\begin{example}
    \label{example:additive:lower:n/2}
    Consider a simple graph consisting of $n/2$ parallel edges, assuming $n$ is even, where each edge is valued at one by both of its endpoint agents.
    For this instance, a subsidy of $n/2$ is required to achieve an envy-freeable orientation since every edge can be allocated to at most one of the endpoint agents and the other agent needs a subsidy of 1.
\end{example}

One may wonder whether the ``Bounded-Subsidy Algorithm'' in \cite{BDNSV20} can achieve the worst-case bound of $n/2$, we show the answer is ``no'' even for simple graphs. The instance is presented in Appendix. 

We fist introduce several definitions. 
Recall that, given graph $G=(V,E)$, $E_{i,j}$ is the set of items between $i,j \in V$. 
Given an orientation $\mathbf{A}$ of $G$, for any $i,j\in N$, $A^j_i$ refers to $A_i\cap E_{i,j}$ with subscript $i$ representing that these items in the bundle of agent $i$ and with superscript $j$ and subscript $i$ representing that these items are from $E_{i,j}$.
We say an orientation $\mathbf{A}$ is \emph{locally EFable} when restricting to $E_{i,j}$ or regarding $i,j$, if and only if $v_i(A^j_i)+v_j(A^i_j) \geq v_i(A^i_j)+v_j(A^j_i)$.
We also use the positive part operator defined as follows: for any $d\in \mathbb{R}$,
    $
    d^+ = \max(0,d).
    $

The main result of this section is a polynomial time algorithm that computes an EFable orientation requiring a total subsidy of at most $n/2$ for multigraphs.

\begin{theorem} 
\label{thm:multi:additive:alg:n/2}
For multigraphs and additive agent valuations, there exists an EF orientation with a subsidy of at most $n/2$, and it can be computed in polynomial time.
\end{theorem}

We now introduce the idea of our algorithm. There are two phases: in Phase 1, we allocate items between adjacent agents in the \emph{reserve graph} $G'$ (will be introduced below), and in Phase 2, we allocate items between non-adjacent agents in $G'$.
The $G'$, a directed graph, is constructed as follows; Each agent forms a vertex in $G'$; For arcs, let each agent $i$ arbitrarily pick an item $(i,j)\in E$ with $v_i(\{(i,j)\})=1$ and add arc $(j,i)$ to $G'$.
Then $G'$ is formed by $k$ connected components $C_1,\ldots,C_k$.

$G'$ is essentially a subgraph of $G$ where $E'\subseteq E$ and every agent $i$ claims exactly one favorite item $(i,j)\in E$ incident to her as the edges in $E'$. By our normalization, $i$ must have value 1 for this item. To indicate the claimant of edges in $E'$, we add directions to edges in $E'$. If agent $i$ claims $(i,j)$ as the favorite, then the edge goes from $j$ to $i$. That is the edges in $E'$ pointing to their claimants. 
Note that if there are two adjacent agents $i$ and $j$ claim the same edge $(i,j)$, then there are two directed edges between vertices $i,j$, one from $i$ to $j$ and the other one from $j$ to $i$.

In Phase 1, we allocate items between adjacent agents in each component $C$. Items between two agents are allocated based on two routines, $\RR(j,i,E_{i,j})$ and $\MAXU(j,i,E_{i,j})$.
The $\RR(j,i,E_{i,j})$ lets agents $i,j$, in turn, pick their most preferred one from the remaining items and agent $j$ picks first; the order is crucial.
The $\MAXU(j,i,E_{i,j})$ allocates each item $e\in E_{i,j}$ to the agent having a larger value for $e$ and if there is a tie, allocate that item to agent $j$.

We will distinguish whether component $C$ contains a cycle with a length of 3 or more. The case where $C$ has such a cycle is relatively easy to handle and we use $\RR$ to allocate items between any pair of adjacent agents in $C$, with a specific picking order determined by $G'$.
For the other case where $C$ has a 2-cycle, composed by vertices $f,g$, we first form out-trees of $f$ and of $g$ (see Figure \ref{fig:2-cycle} for an illustration) and will proceed with the allocation of each out-tree from the corresponding root to all depth 1 vertices and then from depth 1 vertices to all depth 2 vertices, until to the leaves.  
Depending on the partial allocation of $E_{f,g}$, we have different allocation rules.
If $E_{f,g}$ is allocated such that $i,j$ do not envy each other, then we use subroutine $\SONE$ to allocate items in both out-trees.
For the other case, we use subroutine $\STWO$ to allocate the out-tree of the envious agent between $f,g$ while $\SONE$ for that of the other.
Informally, the $\SONE$ is almost identical to $\RR$. The $\STWO$ will classify the successors of the current vertex $i$ and implement either $\RR$ or $\MAXU$ based on the specific classification.
\begin{figure}[H]
       \centering
       \includegraphics[width=0.7\linewidth]{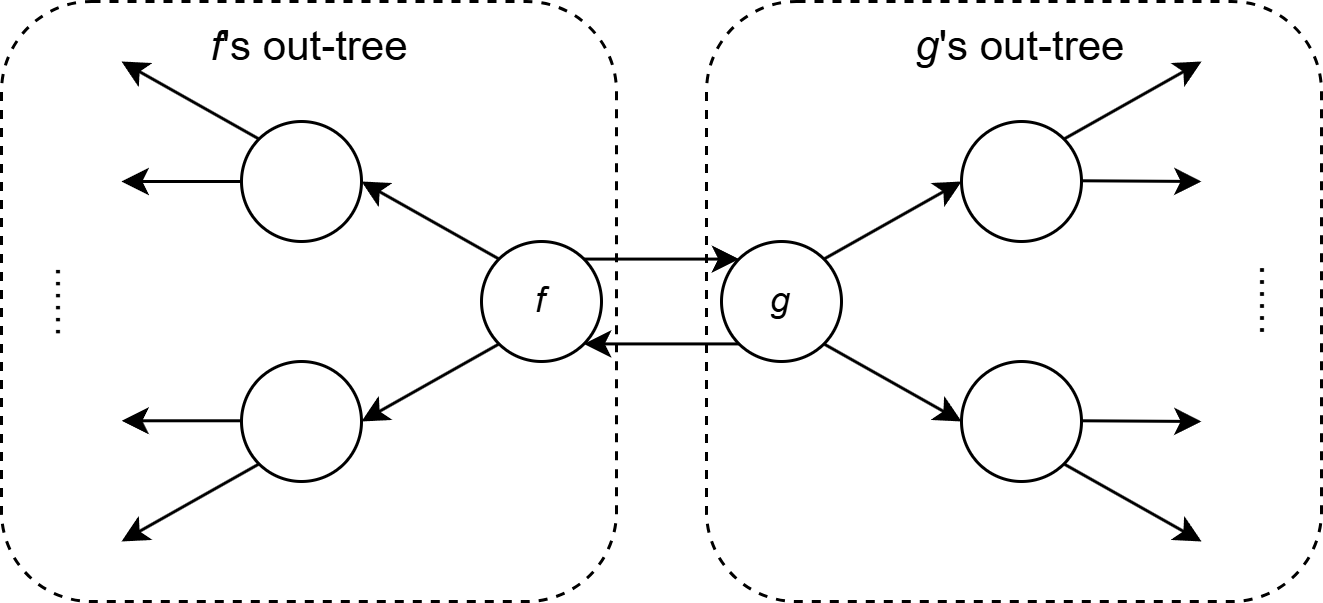}
       \caption{Illustration of the Component with a 2-cycle} 
       \label{fig:2-cycle}
   \end{figure}

\begin{algorithm}[ht!]
	\caption{$\SONE(i,C)$}
	\label{alg:alloc1}
	\begin{algorithmic}[1]
		\REQUIRE Agent $i$ and component $C$.
		\ENSURE An allocation of items belong to the rooted tree starting with $i$. 
            \STATE $\Ch(i)\gets$ the set of $i$'s out-neighbors in $C$, excluding $s$ (if applicable).
            \FOR{every $j \in \Ch(i)$}
                \STATE Let $\RR(j,i,E_{i,j}) = (B_j,B_i)$ and allocate bundles respectively $B_j$ and $B_i$ to agents $j$ and $i$; 
            \ENDFOR
	\end{algorithmic}
    \end{algorithm}

    \begin{algorithm}[H]
	\caption{$\STWO(i,C)$}
	\label{alg:alloc2}
	\begin{algorithmic}[1]
		\REQUIRE Agent $i$ and component $C$.
		\ENSURE An allocation of items belonging to the rooted tree starting with $i$.
            \STATE Initialize $R,Q_1,Q'_1,Q_2,Q_3,Q_4,Q_5 \gets \emptyset$ and $k\gets 0$.
            \STATE $s \gets$ the in-neighbor of $i$ in $C$.
            \STATE $\Ch(i)\gets$ the set of $i$'s out-neighbor in $C$, excluding $s$ (if applicable).
            \STATE For agent $i$, define {$w_i = ( t_s + \dif{i}{s} )^+$}.
            \label{step:sub2-defining-w}
            \FOR{every $j \in \Ch(i)$}
                \STATE Let $\RR(j,i,E_{i,j})=(T^i_j,T^j_i)$ where $T^j_i$ and $T^i_j$ respectively are the virtual bundles corresponding to agent $i$ and $j$.
                {When the underlying agent has a same value for multiple items, the agent always chooses the item that exists in $G'$.}
            \ENDFOR
            \STATE Let $T_s^i\gets A_s^i$ and $k \gets \arg\max_{j\in \Ch(i) \cup \{s\}}(v_i(T^i_j)+ t_j)$\label{step:sub2-defining-k}
            {\color{gray} \% Here $t_j=0$ holds for all $j\in \Ch(i)$.}
            \FOR{$j\in \Ch(i)$ with {$ v_i(\E{i}{j}) < v_i(T^i_k)+t_k$}}\label{step:sub2-checking-rr-k}
                \STATE Let $\MAXU(j,i,E_{i,j})=(S^i_j,S^j_i)$  where $S^j_i$ and $S^i_j$ respectively are another pair of virtual bundles corresponding to agents $i$ and $j$.
                \STATE $R \gets R\cup\{j\}$
            \ENDFOR
            \FOR{$j\in R$ with $0\leq v_j(S^j_i)-v_j(S^i_j)\leq w_i$ }\label{step:sub2-Q1-for-conditon}
                \IF{$Q_1 = \emptyset$}
                    \STATE $Q_1 \gets Q_1 \cup \{j\}$, $A_j \gets A_j\cup S^i_j$, $A_i\gets A_i\cup S^j_i$ and $E\gets E\setminus (S^i_j\cup S^j_i)$ 
                \ELSE
                    \STATE $Q'_1\gets Q'_1\cup \{j\}, A_j \gets A_j\cup T^i_j$, $A_i\gets A_i\cup T^j_i$ and $E\gets E\setminus (T^i_j\cup T^j_i)$ 
                \ENDIF
            \ENDFOR
            \FOR{$j\in R$ with $w_i<v_j(S^j_i)-v_j(S^i_j)$}\label{step:sub2-Q2-for-conditon}
            \STATE $Q_2\gets Q_2\cup \{j\}$, $A_j \gets A_j\cup T^i_j$, $A_i\gets A_i\cup T^j_i$ and $E\gets E\setminus (T^i_j\cup T^j_i)$
            \ENDFOR
            \FOR{$j\in R$ but does not satisfy the conditions of Line \ref{step:sub2-Q1-for-conditon} or \ref{step:sub2-Q2-for-conditon}}
                \STATE $Q_3 \gets Q_3 \cup \{j\}$, $A_j \gets A_j\cup S^i_j$, $A_i\gets A_i\cup S^j_i$ and $E\gets E\setminus (S^i_j\cup S^j_i)$. 
            \ENDFOR
            \STATE \COMMENT{Up to here $R=Q'_1\cup Q_1\cup Q_2 \cup Q_3$ }
            \FOR{$j\notin R$}
            \STATE $A_j \gets A_j\cup T^i_j$, $A_i\gets A_i\cup T^j_i$ and $E\gets E\setminus (T^i_j\cup T^j_i)$. 
            \ENDFOR
            \STATE $S \gets \{j\in \Ch(i) \textnormal{ and } j\notin R: v_i(A_j^i)+v_j(A^j_i) > v_i(A_i)+v_j(A_j^i)\}$. \label{step:sub2-defining-S}
            \IF{$S \neq \emptyset$}
                \STATE $q \gets \arg\max_{j\in S} \ao{i}{j}$ and break ties arbitrarily. Let $Q_4 \gets \{q\}$.
                \STATE  \emph{Swap} $A_i^q$ and $A_q^i$ and update $A_i, A_q$. \label{step:sub2-Q4-swap}
            \ENDIF
            \FOR{$j \in \Ch(i): j\notin Q'_1\cup  Q_1\cup Q_2 \cup Q_3 \cup Q_4$}
                \STATE $Q_5 \gets Q_5 \cup \{j\}$.
            \ENDFOR
            \STATE {$t_i \gets (v_i(A^i_k)+t_k - \ax{i})^+$}\label{step:sub2-update-t}
	\end{algorithmic}
    \end{algorithm}

    We note that for agent $i$, $w_i$ (Line \ref{step:sub2-defining-w} of $\STWO$) is determined before allocating the bundle between $i$ and $i$'s out-neighbor while $t_i$ is determined after all items between $i$ and $i$'s out-neighbor are allocated.
    The high-level idea of $w_i$ is for quantifying the $i$'s out-neighbor into different $Q_{\ell}$'s and for $t_i$, it is used to determine $w_{\ell}$ for every out-neighbor $\ell$ of $i$.
    Indeed, as we will see, $w_i,t_i$ are the upper bounds of the minimum subsidy of envy-freeable allocations for each agent $i$. Moreover, we remark that both $w_i$ and $t_i$ will be updated at most once.

In Phase 2, we allocate items between agents non-adjacent in $G'$. The allocation rule is $\RR$ and we will swap bundles whenever necessary. The formal description of the algorithm is introduced in Algorithm \ref{alg:multigraph-additive}.
The parameters $w_i,t_i$ are used to classify successors in $\STWO$ and we remark that they also act as upper bounds of the final payment $p_i$ for agent $i$, which we will prove below.

        
    \begin{algorithm}[ht!]
        \caption{Multigraph Orientation for Additive Valuations}
	\label{alg:multigraph-additive}
	\begin{algorithmic}[1]
		\REQUIRE An instance $G=(N, E)$ and the valuation functions $\{ v_i \}'s$.
		\ENSURE An orientation $\mathbf{A}=(A_1,\ldots, A_n)$.
            \STATE Initialize $A_i\gets \emptyset, w_i,t_i,p_i \gets 0 $ for all $i\in N$.
            \STATE For each agent $i\in N$, arbitrarily choose an edge $e=(i,j)\in E$ such that $v_i(e)=1$. We construct the \emph{reserve} multigraph $G'=(N,E')$ 
            where vertices are $N$ and there exists an arc $(j,i)\in E'$ from $j$ to $i$ if $e=(i,j)$ was chosen by agent $i$; note if an edge is chosen by both agents of $i$ and $j$, we add both arcs of $(j,i)$ and $(i,j)$ to $G'$.\label{step:mainAlg-preprocess}
            \STATE Let $\mathcal{C}=\{C_1,\ldots,C_k\}$ be the $k$ connected components of $G'$.
            
            {\color{gray} \% Phases 1: allocate items between adjacent agents in $G'$}
            \FOR{each component $C \in \mathcal{C}$}
                \IF{$C$ has a cycle of length 3 or more}
                    \FOR{each arc $(i,j)$ in $C$}
                        \STATE Let $\RR(j,i, E_{i,j}) = (B_j,B_i)$ and respectively allocate $B_j$ and $B_i$ to agents $j$ and $i$.
                    \ENDFOR
                \ELSE
                    \STATE Suppose agents $f,g$ are the (only) two agents
                    such that arcs $(f,g)$ and $(g,f)$ exist in $C$ and then execute $\RR(g,f,E_{f,g})$. 
                    If the resulting partial allocation is not locally EFable when restricting to $E_{f,g}$,
                    we \emph{swap} their partial bundles restricted to $E_{f,g}$.
                \label{step:mainAlg-f-g-swap}
                    \IF{both $f$ and $g$ are envy-free when restricting to the allocation of $E_{f,g}$}
                    \STATE Start from $f$ and form a out-tree (excluding $g$) with root $f$, and execute $\SONE(\cdot)$ from $f$ to the leaves.
                    \STATE Start from $g$ and form a out-tree (excluding $f$) with root $g$, and execute $\SONE(\cdot)$ from $g$ to the leaves.
                    \ELSE
                    \STATE Suppose that agent $g$ \emph{does not envy} $f$ when restricting to the allocation of $E_{f,g}$; {\color{gray} \% Local EFability guarantees $g$}
                    \STATE Start from $f$ and form a out-tree (excluding $g$) with root $f$, and execute $\STWO(\cdot)$ from $f$ to the leaves.\label{step:mainAlg-Sub2-for-f}
                    \STATE Start from $g$ and form a out-tree (excluding $g$) with root $g$, and execute $\SONE(\cdot)$ from $g$ to the leaves.
                    \ENDIF
                \ENDIF
            \ENDFOR 
            
            {\color{gray} \% Phases 2: allocate items between non-adjacent agents in $G'$}
            \STATE For every $i\in N$, $p_i\gets t_i$.
            \FOR{each set of edges $\E{i}{j}$ not allocated}\label{step:mainAlg-start-Phase2}
                \STATE Execute $\RR(j,i,E_{i,j})$. If the resulting partial allocation is not locally envy-freeable when restricting to $E_{i,j}$, swap their partial bundles restricted to $E_{i,j}$.\label{step:mainAlg-last-swap}
                \STATE Let $j$ be the agent such that $\ndif{j}{i} \geq 0$. {\color{gray} \% Local EFability guarantees $j$}
                \IF{$\ao{i}{j}+p_j > \ax{i}+p_i $}\label{step:mainAlg-last-if}
                    \STATE $p_j \gets (\ax{i}+p_i - \ao{i}{j})^+$.\label{step:mainAlg-last-adjust-p}
                \ENDIF
            \ENDFOR
	\end{algorithmic}
    \end{algorithm}

    Before proving the main results, we first present observations and propositions regarding the algorithm and subroutines.
   
    \begin{observation}
        For every pair of agents $i$ and $j$, if they are adjacent in $G$, then bundles $A_i^j$ and $A_j^i$ are determined by either $\RR$ or $\MAXU$.
    \end{observation}
    This observation directly follows from the way of allocating $E_{i,j}$ in the algorithm. We remark that necessary swaps such as Line \ref{step:sub2-Q4-swap} in $\STWO$ and Lines \ref{step:mainAlg-f-g-swap} and \ref{step:mainAlg-last-swap} in Algorithm \ref{alg:multigraph-additive}, are made to ensure local EFablility. 

    \begin{proposition}\label{prop::rr}
        Suppose $\RR(j,i,E_{i,j})=(B_j,B_i)$, then
        $v_j(B_j)\geq v_j(B_i)$ and $v_i(B_j)-v_i(B_i)\leq 1$ hold.
        If respectively assigning $B_j$ and $B_i$ to agents $j$ and $i$ does not result in a locally envy-freeable allocation, then
        $v_j(B_j)-v_j(B_i)\leq 1$.
    \end{proposition}
    
    The proposition indicates that once the allocation of $E_{i,j}$ is determined by $\RR$ in the algorithm, the envy of any of these two agents towards the other is at most 1. 
    
    \begin{proposition}\label{prop::greedy}
        Suppose $\MAXU(j,i,E_{i,j})=(B_j,B_i)$ and there exists an item $e$ valued at 1 by agent $j$, then $v_j(B_j)\geq 1$.
    \end{proposition}
    
    The proposition states that if the allocation of $E_{i,j}$ is determined by $\MAXU$ and the advantage agent has a value of 1 for some item in $E_{i,j}$, then the advantage agent receives a value of at least 1 from items $E_{i,j}$.
    The following proposition indicates the bounds of $w_i$ and $t_i$ (if applicable).
    \begin{proposition} \label{prop::bound}
    For agent $i$, if the items between $i$ and $i$'s direct successors in the connected component $C$ of $G'$ are allocated by $\STWO$, then in the same component, $i$ has a direct predecessor $s$.
    Moreover, $v_i(A^s_i)+w_i\geq 1$ and after updating $t_i$ in Line \ref{step:sub2-update-t} of $\STWO$, $t_i\leq w_i$.
    \end{proposition}
    \begin{proof}
        The fact that $i$ has a predecessor $s$ in the same component directly follows from Algorithm \ref{alg:multigraph-additive}. Since arc $(s,i)$ exists in $C$, according to Line \ref{step:mainAlg-preprocess} of Algorithm \ref{alg:multigraph-additive},
        there exists $e\in E_{s,i}$ such that $v_i(e)=1$.
        Thus, $\max(\ai{i}{s},\ao{i}{s}) \geq 1$ holds as $\{ A^s_i, A^i_s \}$ is a 2-partition of $E_{s,i}$. Then we have the following,
        \begin{align*}
            w_i + v_i(A^s_i) & = \left(t_s + \dif{i}{s} \right)^+ +  \ai{i}{s} \\
            &\geq \max(\ai{i}{s},\ao{i}{s}) \geq 1,
        \end{align*}
        where the first inequality transition is due to $t_s\geq 0$.

        For the bound between $t_i$ and $w_i$, if $t_i=0$, then $t_i\leq w_i$ holds trivially. 
        For the case where $t_i>0$, we have $t_i=t_k + \ao{i}{k} - \ax{i}$ for some $k$ belonging to $\Ch(i)\cup \{s\}$; recall that $\Ch(i)$ is the set of direct successors of $i$ in $C$ (excluding $s$, which we must take into account when determining $k$).
        If $k=s$, by the definitions of $t_i,w_i$, we directly have $t_i\leq w_i$.
        If $k\neq s$, at the moment when computing $t_i$ at Line \ref{step:sub2-update-t}, $t_k=0$. By $k\neq s$, we have $A^k_i\cup A^s_i \subseteq A_i$ and the following holds,
        \begin{align*}
            t_i=t_k+ \ao{i}{k} - \ax{i}& \leq \ao{i}{k} -\ai{i}{k} - \ai{i}{s} \\
            &\leq 1 - \ai{i}{s} \leq w_i.
        \end{align*}
        The first inequality comes from $t_k=0$ and $A^k_i\cup A^s_i \subseteq A_i$. For the second one, since $k\in Q_4\cup Q_5$, the 2-partition $\{A^k_i,A^i_k\}$ of $E_{i,k}$ is determined by $\RR(k,i,E_{i,k})$, and hence by Proposition \ref{prop::rr}, $\ao{i}{k} -\ai{i}{k}\leq 1$.
        The third one is due to $v_i(A^s_i)+w_i\geq 1$.
    \end{proof}

    In what follows, we present a crucial lemma, helpful for bounding the total required subsidy. stating that for some agent $i$, if $\STWO$ is used to allocate items between $i$ and $i$'s direct successors $\Ch(i)$ in 
    some component $C$ of $G'$, then $t_i+\sum_{j\in \Ch(i)}w_j \leq w_i$. Note that if $i$'s predecessor is also her successor, we always exclude the predecessor from $\Ch(i)$.
    \begin{lemma}\label{lem::t-sum-w-bound}
            If $\STWO$ is used to allocate items between $i$ and $i$'s direct successors $\Ch(i)$ in a component $C$ of $G'$, then $t_i+\sum_{j\in \Ch(i)}w_j \leq w_i$.
    \end{lemma}
    To prove this lemma, we establish the key properties guaranteed by $\STWO$ for $i$'s direct successors.
    The proof of this lemma involves a detailed discussion and is deferred to Appendix.

\begin{lemma}\label{lem::upper-bound-of-t}
        For each component $C$ in $G'$, let $N_C$ be the set of vertices in $C$. Then it holds that $\sum_{j\in N_C} t_j \leq 1$ at termination of Algorithm \ref{alg:multigraph-additive}. 
    \end{lemma}
    \begin{proof}[Proof of Lemma \ref{lem::upper-bound-of-t}]
        First, we remark that $t_j$'s do not get updated on or after Line \ref{step:mainAlg-start-Phase2} of Algorithm \ref{alg:multigraph-additive}. 
        The algorithm only updates $t_j$'s in $\STWO$ and hence if $t_j>0$ for some $j\in N_C$, then $\STWO$ is used to allocate the items adjacent to agents in $C$, 
        Let $f\in N_C$ be the root of the out-tree. Then according to Proposition \ref{prop::bound}, for any $j$ in the rooted tree of $f$, $t_j \leq w_j$ holds.
        Then by applying Lemma \ref{lem::t-sum-w-bound} from the leaves to the direct successors of root $f$, we have
        \begin{align*}
            \sum_{j\in N_C} t_j &= t_f + \sum_{j\in N_C\setminus \{f\} } t_j \leq t_f+ \sum_{j\in N_C\setminus \{f\} } w_j \\
            &\leq t_f + \sum_{j\in \Ch(f)} w_j \leq w_f,
        \end{align*}
        where the last two inequality transitions follow from Lemma \ref{lem::t-sum-w-bound}.
        Note that in the proof of Lemma \ref{lem::t-sum-w-bound}, we have proved $w_f\leq 1$ when $f$ is the root and therefore complete the proof of the lemma.  
    \end{proof}

    At the end of Algorithm \ref{alg:multigraph-additive}, we will let each agent $j$ receive bundle $A_j$ and payment $p_j$. According to Phase 2 of Algorithm \ref{alg:multigraph-additive}, for any $j\in N$, $p_j\leq t_j$ holds. Thus, the total subsidy $\sum_{j\in N} p_j \leq \sum_{j\in N} t_j \leq \frac{n}{2}$ where the last inequality follows from Lemma \ref{lem::upper-bound-of-t} and the fact that every component has at least two vertices.
    Hence, the remaining is to prove that the orientation with payment $\{\mathbf{A}, \mathbf{p}\}$ is envy-free.

    \begin{lemma}\label{lem::efable-phase1}
        At the moment right before Phase 2, no agent envies their out- and in-neighbours in $G'$ if $t_i$ is the payment to agent $i$.
    \end{lemma}


    \begin{lemma}\label{lem::final-efable}
        At the termination of Algorithm \ref{alg:multigraph-additive}, the allocation with payment $\{\mathbf{A}, \mathbf{p}\}$ is envy-free.
    \end{lemma}
    The proofs for the above two lemmas involve detailed analysis of $t_i,w_i$ for all $i\in N$. The formal proofs can be found in the appendix.


\section{Improved Bounds for Simple Graphs}
\label{sec:simple:monotone}

Finally, we improve Theorem \ref{thm:multi:monotone:alg:n-1} for simple graphs, where the valuations can be arbitrarily monotone. We start with the lower bound.

\paragraph{Lower bound.}  The following example gives an instance for which any EF orientation requires a subsidy of at least $n-2$. 
\begin{example} Consider a simple graph $G$ with two connected components $C_1$ and $ C_2$,  where:
\begin{itemize}
\label{expl:monotone:n-2}
    \item $C_1=(V_{C_1},E_{C_1})$ is a single edge connecting two vertices denoted $V_{C_1}=\{1,2 \}$, and both agents (vertices 1 and 2) value the edge incident to them at 1.  
    \item  $C_2=(V_{C_2},E_{C_2})$ is a complete graph  $K_{n-2}$ with $V_{C_2} = \{3,\ldots,n \}$. The value of each agent depends solely on the number of edges $S$ they receive, and is given by the valuation 
    $$
    v_i(S) =
    \begin{cases}
	1 \ \ \ &\text{if}  \ \ |S| = n-2 \ \\
	0 \ \ \ \ \quad  &\text{otherwise.}
    \end{cases}
    $$
\end{itemize}
In any orientation of $C_1$, one agent (say agent 1) receives the edge, leaving the other agent (agent 2) with value zero. To ensure an EF orientation, agent 2 must receive a subsidy of 1. For every agent $i \in V_{C_2}$, their value is 0 unless they receive all the edges incident to them. Therefore, in any orientation, at most one agent in $V_{C_2}$ can achieve a value of 1, while the others receive zero. Since agent 2 from $C_1$ receives a subsidy of 1, 
the agents in $V_{C_2}$ who have zero value (there are $n-3$ of them) must each receive a subsidy of 1 to maintain EF.
Thus, the minimum subsidy needed for an EF orientation in this instance is $n-2$.
\end{example}

\paragraph{Upper bound.}
We now show that a subsidy of $n-2$ is sufficient to guarantee EF. 
At a high level, we find an EF orientation for which at least two agents receive zero subsidy. The choice of such agents depends on both the structure of $G$ and the  agents' valuations.

\begin{algorithm}[ht!]
        \caption{Simple Graph Orientation with Additive Valuations}
	\label{alg:subsidy-monotone-simple}
	\begin{algorithmic}[1]
 \renewcommand{\algorithmicensure}{\textbf{Output:}}
		\REQUIRE An instance  $G=(N, E)$ and valuation function $\{ v_i \}$'s.
		\ENSURE EF orientation $\mathbf{A} = (A_1,\ldots,A_n)$ and payments $\mathbf{p}=(p_1,\ldots, p_n)$.
  \STATE Initialize $A_i=\emptyset$ for each $i\in N$
        \STATE Let $t = \max_{i\neq j} v_i(\{(i,j)\})$, and let $i'$ and $j'$ be the two vertices for which $v_{i'}(\{(i',j')\})=t$. Choose any $j^*\in N\setminus \{i',j'\}$, and allocate $j^*$ all edges incident to her, and update $A_{j^*}$.
        \STATE\label{alg:unalloc} For each unallocated edge $e$, allocate $e$ to the incident agent who values it the most, breaking ties arbitrarily, and update the corresponding $A_i$'s.
        \STATE For each $i\in N$, set the subsidy payments as $p_i=\max\{ t - v_i(A_i),0\}$.
	\end{algorithmic}
\end{algorithm}
\begin{theorem}
\label{thm:simple:monotone:alg:n-2}
For simple graphs with monotone agent valuations, there always exists an EF orientation with a subsidy of at most $n-2$, and it can be computed in polynomial time using value queries.
\end{theorem}
\begin{proof}
We show that Algorithm~\ref{alg:subsidy-monotone-simple} returns an EF orientation with total subsidy of at most $n-2$. As the maximum marginal value of each edge is at most 1, we see that $p_i\leq t \leq 1$.

Consider any pair of agents $i,j\in N\setminus j^*$. We show that $i$ is EF towards $j$. Note that the claim holds immediately when $i$ and $j$ are not adjacent or the edge $(i,j)$ is oriented toward $i$, since $v_i(A_i)+p_i\geq t \geq p_j$. Thus, we may only focus on the case where $i$ and $j$ are adjacent and the edge $(i,j)$ is oriented towards $j$. If $p_j=0$, we have $v_i(A_i)+p_i\geq t \geq v_i(\{(i,j)\})=v_i(A_j)=v_i(A_j)+p_j$. Furthermore, if $p_j>0$, meaning that  $t>v_j(A_j)$, we have 
\begin{align*}
    v_i(A_j) + p_j &= v_i(\{(i,j)\}) + t-v_j(A_j) \\ 
    &\leq v_i(\{(i,j)\}) + t - v_j(\{(i,j)\}) \\ 
    &\leq t \leq v_i(A_i) +p_i,
\end{align*}
where the second inequality follows from  $v_j(\{(i,j)\}) \geq  v_i(\{(i,j)\})$, which follows from the fact that $(i,j) \in A_j$ and Line~\ref{alg:unalloc}  was executed.

Next, we show that no agent envies $j^*$, and that $j^*$ is EF towards any other agent. Since $j^*$ receives all the edges incident to her, we have $v_{j^*}(A_{j^*})\geq 1 \geq v_{j^*} (A_i)+p_i = p_i$ for any $i\in N\setminus j^*$, which means that $j^*$ does not envy any other agent. Furthermore,  for any $i\in V\setminus j^*$, we have $v_i(A_i)+p_i \geq t \geq \max\limits_{h\neq j^* } v_h(\{(h,j^*)\}) \geq v_i(\{(i,j^*)\})= v_i(A_{j^*})$, noting $p_{j^*}=0$, we have that any agent $i\in N\setminus j^*$ does not envy $j^*$. Hence,  $(\mathbf{A},\mathbf{p})$ is indeed EF. 

As for the subsidy bound, we now show that at least one other vertex other than $j^*$ receives a subsidy of zero. Recall that $t=\max_{i\neq j}v_i(\{(i,j)\})$, and $v_{i'}(\{(i',j')\})=t$. Since $j^*\not\in \{i', j'\}$, one of two agents, say $i'$,   receives the edge $(i',j')$, and thus has value at least $t$ (by Line~\ref{alg:unalloc}), meaning that they get zero subsidy. Therefore, at least two agents ($j^*$ and $i'$ ) receive zero subsidy, and since $p_i\leq 1$ for each $i\in N$, the total subsidy is at most $n-2$.


Finally, we observe that Algorithm~\ref{alg:subsidy-monotone-simple} runs in $O(n^2)$ time and makes at most $O(n^2)$ value queries. 
\end{proof}

\section{Conclusion}
In this paper, we provide the (tight) bounds of the minimum subsidy to compensate agents in order to achieve an envy-free allocation in the graph orientation problem. 
Our paper uncovers several interesting problems for subsequent research.
Firstly, there is a constant additive gap between the upper and lower bounds in the general case of multigraph and monotone valuations. 
Closing this gap is an intriguing theoretical problem.
Secondly, we have been focusing on the allocation of goods and the mirror problem of chores/mixed manna has not been studied yet.
Finally, the optimal bound of subsidy in the unrestricted setting (outside the scope of graph orientations) when agents have arbitrary monotone valuations \cite{BDNSV20,DBLP:conf/ijcai/BarmanKNS22,DBLP:conf/aaai/KawaseMSTY24} remains unknown. 


\newpage
\bibliographystyle{alpha}
\bibliography{sample-base}

\onecolumn
\appendix



\newpage

\appendix
\section{Missing Proofs from Section~\ref{sec:Monotone} }

\subsection{Proof of Lemma~\ref{lem:mono-multi-efable}}
\begin{proof}[Proof of Lemma~\ref{lem:mono-multi-efable}]
      By Theorem~\ref{thm:condition-envyfreeable}, it suffices to show that in the envy-graph $D_{\mathbf{A}}$, every cycle has a non-positive  weight. Consider any directed cycle $\mathcal{C}=\{1,...,h,1\}$ in $D_{\mathbf{A}}$. By Lemma~\ref{lem:envy-cycle}, there is at least one agent in $\mathcal{C}$ who do not envy their out-neighbor. Let $\{r_1, r_2,\ldots,r_s\} \subseteq  \mathcal{C}$ with $r_1<r_2<\cdots<r_s$  be the set of agents in $\mathcal{C}$, each of whom does not envy their out-neighbour in $\mathcal{C}$. For notational convenience,  we  define $r_0=0$  and $r_{s+1} = h$, as well as identify $h+1$ with $1$. Note that $j$ and $j+1$ may not be adjacent in $G$ and if so, $E_{j,j+1}, P^{j+1}_{j,j+1},P^{j}_{j,j+1}$ are defined as $\emptyset$. 
    With these notations in hand, we may write the weight of $\mathcal{C}$
     as $w(\mathcal{C}) = \sum_{k=0}^s\sum_{j=r_k+1}^{r_{k+1}} w(j,j+1)$. We have the following 
     \begin{align}
         \sum_{j=r_k+1}^{r_{k+1}} w(j,j+1) &=  \sum_{j=r_k+1}^{r_{k+1}} v_j(A_{j+1})-v_j(A_{j}) \nonumber \\ 
         &=\sum_{j=r_k+1}^{r_{k+1}} v_j( P^{j+1}_{j,j+1})-v_j(A_{j})\nonumber \\ 
         &= -v_{r_k+1}(A_{r_k+1}) + v_{r_{k+1}}(P^{r_{k+1}+1}_{r_{k+1},r_{k+1}+1})+\sum_{j=r_k+1}^{r_{k+1}-1} v_j( P^{j+1}_{j,j+1})- \sum_{j=r_k+2}^{r_{k+1}}v_j(A_{j}) \nonumber \\ 
         &=  -v_{r_k+1}(A_{r_k+1}) + v_{r_{k+1}}(P^{r_{k+1}+1}_{r_{k+1},r_{k+1}+1})+\sum_{j=r_k+1}^{r_{k+1}-1} v_j( P^{j+1}_{j,j+1})- \sum_{j=r_k+1}^{r_{k+1}-1}v_{j+1}(A_{j+1}) \nonumber \\ 
          &\leq -v_{r_k+1}(A_{r_k+1}) + v_{r_{k+1}}(P^{r_{k+1}+1}_{r_{k+1},r_{k+1}+1})+\sum_{j=r_k+1}^{r_{k+1}-1} b_j - b_{j+1} \nonumber \\
            &= -v_{r_k+1}(A_{r_k+1}) + v_{r_{k+1}}(P^{r_{k+1}+1}_{r_{k+1},r_{k+1}+1})+b_{r_k+1} - b_{r_{k+1}} \label{eq:weight}
     \end{align}
The inequality transition above holds since for any $r_k +1 \leq j < r_{k+1}$ (if any), agent $j$ envies her out-neighbour $j+1$ on $\mathcal{C}$, which by Line~\ref{alg:line:temp}, implies that both $j,j+1$ prefers $P^{j+1}_{j,j+1}$ to $P^{j}_{j,j+1}$ and moreover $v_j(P^{j+1}_{j,j+1}) - b_j \leq v_{j+1}(P^{j+1}_{j,j+1}) - b_{j+1}$, equivalently  $v_j(P^{j+1}_{j,j+1}) - v_{j+1}(P^{j+1}_{j,j+1}) \leq b_j  - b_{j+1}$. Thus,  $v_j( P^{j+1}_{j,j+1}) - v_{j+1}(A_{j+1}) \leq v_j(P^{j+1}_{j,j+1}) - v_{j+1}(P^{j+1}_{j,j+1})\leq b_j  - b_{j+1}$. 

Using inequality~\ref{eq:weight}, along with the fact that $w(\mathcal{C}) = \sum_{k=0}^s\sum_{j=r_k+1}^{r_{k+1}} w(j,j+1)$. We get a bound on the weight of the cycle,
\begin{align}
    w(\mathcal{C}) \leq \sum_{k=0}^s -v_{r_k+1}(A_{r_k+1}) + v_{r_{k+1}}(P^{r_{k+1}+1}_{r_{k+1},r_{k+1}+1})+b_{r_k+1} - b_{r_{k+1}} \label{eq:cycle:bound}
\end{align}

    Then for any $k \in \{0,\ldots,s\}$, we define an expression $F_k$ as follows,
    $$
    F_k=  - v_{r_{k+1}+1}(A_{r_{k+1}+1})+ v_{r_{k+1}}(P^{r_{k+1}+1}_{r_{k+1},r_{k+1}+1})  + b_{r_{k+1}+1}-b_{r_{k+1}},
    $$
    We remark that the right hand side of inequality~(\ref{eq:cycle:bound}) is equal to $\sum_{k=0}^s F_k$ as it is simply a rearrangement of the terms in each sum. Thus, we have $w(\mathcal{C}) \leq \sum_{k=0}^s F_k$, which means it suffices to show that $F_k\leq 0$ for each $k \in \{0,\ldots,s\}$. 

    We now show that, for each $k$, $F_k\leq 0$. If $v_{r_{k+1}}(P^{r_{k+1}+1}_{r_{k+1},r_{k+1}+1}) \leq b_{r_{k+1}}$, it suffices to prove $v_{r_{k+1}+1}(A_{r_{k+1}+1})\geq b_{r_{k+1}+1}$. This holds because each agent $i \in N$ is guaranteed a value of at least $b_i$ by monotonicity. 
    
    On the other hand, if $v_{r_{k+1}}(P^{r_{k+1}+1}_{r_{k+1},r_{k+1}+1}) > b_{r_{k+1}}$,  then agents $r_{k+1}$ and $r_{k+1}+1$ are adjacent in $G$. Moreover, agent $r_{k+1}$ prefers $P^{r_{k+1}+1}_{r_{k+1},r_{k+1}+1}$ to $P^{r_{k+1}}_{r_{k+1},r_{k+1}+1}$, and by Line~\ref{alg:line:temp}, we have
    $$
    v_{r_{k+1}}(P^{r_{k+1}+1}_{r_{k+1},r_{k+1}+1}) - b_{r_{k+1}} \leq v_{r_{k+1}+1}(P^{r_{k+1}+1}_{r_{k+1},r_{k+1}+1}) - b_{r_{k+1}+1}.
    $$
    Combining the above inequality and the fact that $P^{r_{k+1}+1}_{r_{k+1},r_{k+1}+1}\subseteq A_{r_{k+1}+1}$, we have
    $$
    \begin{aligned}
    v_{r_{k+1}}(P^{r_{k+1}+1}_{r_{k+1},r_{k+1}+1}) - v_{r_{k+1}+1}(A_{r_{k+1}+1}) &\leq v_{r_{k+1}}(P^{r_{k+1}+1}_{r_{k+1},r_{k+1}+1}) - v_{r_{k+1}+1}(P^{r_{k+1}+1}_{r_{k+1},r_{k+1}+1}) \\
    & \leq b_{r_{k+1}} - b_{r_{k+1}+1},
        \end{aligned}
    $$
    equivalently  $F_k\leq 0$. Therefore, $w(\mathcal{C})\leq 0$, as needed to show.
\end{proof}

\subsection{Proof of Lemma~\ref{lem:weight-length-mono}}

\begin{proof}[Proof of Lemma~\ref{lem:weight-length-mono}]
Consider the maximum weight path $\cP=\{1,2,\ldots,h-1,h\}$. Without loss of generality, we may assume $h-1$ envies $h$, i.e., $w(h-1,h)>0$; otherwise, we can iteratively remove the last vertex from the path without changing its weight until this property holds.
Let $\{r_1,r_2,\ldots,r_s\}\subseteq [h]$ with $r_1<r_2<\cdots<r_s$ be the set of agents, each of whom does not envy their out-neighbour on path $\cP$. For notational convenience, we also define $r_0=0$. Similarly to the proof of Lemma~\ref{lem:mono-multi-efable}, if $j$ and $j+1$ are not adjacent in $G$, then $E_{j,j+1}, P^j_{j,j+1}, P^{j+1}_{j,j+1}$ are defined as $\emptyset$.
For the case when $j$ envies $j+1$, it must be the case that $j$ and $j+1$ are adjacent in $G$ and both prefer $P^{j+1}_{j,j+1}$ to $P^j_{j,j+1}$.

The weight of $\cP$ satisfies,
    $$
    \begin{aligned}
        w(\cP) & = \sum_{k=0}^{s-1}\sum_{j=r_k+1}^{r_{k+1}}w(j,j+1) + \sum_{j=r_s+1}^{h-1}w(j,j+1) \\
        & \leq \sum_{k=0}^{s-1}-v_{r_k+1}(A_{r_k+1}) +b_{r_k+1} - b_{r_{k+1}} +v_{r_{k+1}}(P^{r_{k+1}+1}_{r_{k+1},r_{k+1}+1})+ \sum_{j=r_s+1}^{h-1}w(j,j+1). \\
    \end{aligned}
    $$
    For every $k\in \{0,\ldots,s-1\}$, we define
    $$
    F_k = -b_{r_{k+1}}+v_{r_{k+1}}(P^{r_{k+1}+1}_{r_{k+1},r_{k+1}+1}) - v_{r_{k+1}+1}(A_{r_{k+1}+1}) + b_{r_{k+1}+1}.
    $$
    Given that 
    $$
    \sum_{j=r_s+1}^{h-1}w(j,j+1) \leq -v_{r_s+1}(A_{r_s+1}) + b_{r_s+1} - b_{h-1}+v_{h-1}(P^h_{h-1,h}),
    $$
    we have 
    $$
    w(\cP) \leq \left( -v_1(A_1)+b_1-b_{h-1}+v_{h-1}(P^h_{h-1,h}) \right)+  \sum_{k=0}^{s-1}F_k.
    $$
    By similar argument to that in the proof of Lemma~\ref{lem:mono-multi-efable}, for any $k\in \{0,\ldots,s-1\}$, we have $F_k\leq 0$. Moreover, since agent 1 receives a bundle from each of her neighbours in $G$, we have $v_1(A_1) \geq b_1$ due to the monotonicity of $v_1$. 
    Thus, to prove the statement, it suffices to show $v_{h-1}(P^h_{h-1,h})-b_{h-1} \leq 1$. Since $h-1$ envies $h$ in $\mathbf{A}$, agents $h-1$ and $h$ are adjacent in $G$ and both of them prefer $P^h_{h-1,h}$ to $P^{h-1}_{h-1,h}$.

    If agent $h-1$ receives bundle $P^{h-1}_{h-1,h}$ after the envy-cycle elimination, then she is EF1 when restricting to the partial allocation where agent $h-1$ and $h$ receives $P^{h-1}_{h-1,h}$ and $P^{h}_{h-1,h}$ respectively. By EF1, we have
    $$
    \begin{aligned}
    &v_{h-1}(P^{h-1}_{h-1,h})\geq \min_{g\in P^{h}_{h-1,h}} v_{h-1}(P^{h}_{h-1,h}\setminus \{g\}) \\
    &= v_{h-1}(P^h_{h-1,h}) - \left(v_{h-1}(P^h_{h-1,h}) - \min_{g\in P^{h}_{h-1,h}} v_{h-1}(P^{h}_{h-1,h}\setminus \{g\}) \right) \\
    & \geq v_{h-1}(P^h_{h-1,h})  -1,
    \end{aligned}
    $$
    where the last inequality transition is because the maximum marginal value of adding an item to a bundle is at most 1.
    Since agent $h-1$ prefers $P^h_{h-1,h}$ to $P^{h-1}_{h-1,h}$, when computing $b_{h-1}$, bundle $P^{h-1}_{h-1,h}$ has been taken into account and thus $b_{h-1} \geq v_{h-1}(P^{h-1}_{h-1,h})$.
    Therefore, $b_{h-1} \geq v_{h-1}(P^h_{h-1,h})-1$ holds.

    If agent $h-1$ receives $P^h_{h-1,h}$ right after the envy-cycle elimination, at that moment, bundle $P^{h-1}_{h-1,h}$ was temporarily allocated to agent $h$ who is EF1 when restricting to this partial allocation. Then, we have
    $$
    \begin{aligned} 
    v_h(P^{h-1}_{h-1,h}) &\geq \min_{g\in P^h_{h-1,h}} v_h(P^h_{h-1,h}\setminus \{g\}) \\
    & = v_h(P^h_{h-1,h}) - \left( v_h(P^h_{h-1,h}) - \min_{g\in P^h_{h-1,h}} v_h(P^h_{h-1,h}\setminus \{g\}) \right) \\
    & \geq v_h(P^h_{h-1,h}) - 1,
    \end{aligned}
    $$
    where the last inequality is due to that the maximum marginal value of adding an item to a bundle is at most 1.
    Note that by Line~\ref{alg:line:temp}, $P^h_{h-1,h}$ permanently allocated to agent $h$, and according to the allocation rule, we have
    $$
    v_h(P^h_{h-1,h}) - b_h \geq v_{h-1}(P^h_{h-1,h}) - b_{h-1}.
    $$
    Since agent $h$ prefers $P^h_{h-1,h}$ to $P^{h-1}_{h-1,h}$, we have $b_h\geq v_h(P^{h-1}_{h-1,h})$. Combining the above inequalities, we have
    \begin{align*}
        v_{h-1}(P^h_{h-1,h}) - b_{h-1} &\leq v_h(P^h_{h-1,h}) - b_h \leq v_h(P^h_{h-1,h})- v_h(P^{h-1}_{h-1,h}) \leq 1,
    \end{align*}
    completing the proof.
\end{proof}

\subsection{Proof of Theorem~\ref{thm:multi:monotone:alg:n-1}}

\begin{proof}[Proof of Theorem~\ref{thm:multi:monotone:alg:n-1}]

    As shown in Lemmas \ref{lem:mono-multi-efable} and \ref{lem:weight-length-mono}, Algorithm~\ref{alg:subsidy-monotonoe} returns an EF allocation where the subsidy to each agent is at most one. Since it is possible to uniformly decrease subsidy to all agents while maintaining envy-freeness, there is at least one agent that receives subsidy of zero. Therefore, the $n-1$ bound holds. 

    Algorithm~\ref{alg:subsidy-monotonoe} uses the envy-cycle elimination algorithm of Lipton et al. \cite{DBLP:conf/sigecom/LiptonMMS04}, which is a polynomial-time algorithm, as a subroutine.  The envy-cycle elimination algorithm is called $O(n^2)$ times, and the rest of the  algorithm clearly runs in polynomial time. Thus, Algorithm~\ref{alg:subsidy-monotonoe} is a polynomial-time algorithm. 
\end{proof}

\newpage

\section{Missing Materials from Section \ref{sec::n/2}}
\subsection{Bounded-Subsidy Algorithm in \cite{BDNSV20} needs subsidy more than $n/2$}
\begin{example}
    Consider an example where there are 5 agents $\{1,\ldots, 5 \}$ located on a path graph. There are four items: $e_1 = (1,2)$, $e_2 = (2,3)$, $e_3 = (3,4)$ and $e_4 = (4,5)$.
    For every agent, the values of her relevant items are: $v_1(e_1) = 1$ and $v_2(e_1) = 1, v_2(e_2) = \epsilon^2$ and $v_3(e_2) = 1,v_3(e_3) = 1-\epsilon$ and $v_4(e_3) = \epsilon,v_4(e_4) = 1$ and $v_5(e_4) = 1$ where $\epsilon>0$ is arbirtarily small.
    The ``Bounded-Subsidy Algorithm'' executes one round and returns a matching with total weight $3+\epsilon$.
    There are two possible maximum weight matchings $\mu_1$ and $\mu_2$; (i) $\mu_1$ with $\mu_1 (1) = e_1$, $\mu_1(2) = \emptyset$, $\mu_1(3) = e_2$, $\mu_1(4) = e_3$ and $\mu_1(5) = e_4$, 
    and (ii) $\mu_2$ with $\mu_2(1) = \emptyset$, $ \mu_2(2) = e_1$, $\mu_2(3) = e_2$, $\mu_2(4) = e_3$ and $\mu_2(5) = e_4$.
    For computing the minimum subsidy, let us look at the envy-graph $G_{\mu_1}$ corresponding to allocation $ \{ \mu_1(i) \}$ where the maximum weight of path beginning at each vertex $i$ in the envy graph lower bounds the payment to agent $i$.
    Accordingly, it is not hard to verify that $p_2 \geq 1-2\epsilon + \epsilon^2$, $p_3 \geq 1-2\epsilon$ and $p_4 \geq 1-\epsilon$, and thus the total subsidy required is at least $3 > \frac{5}{2}$ when $\epsilon \rightarrow 0$.
    Similarly, for matching $\mu_2$, we have $p_1 \geq 1$, $p_3 \geq 1-2\epsilon$ and $p_4 \geq 1-\epsilon$, and thus the total subsidy required is at least $3 > \frac{5}{2}$ when $\epsilon \rightarrow 0$. 
\end{example}

\subsection{Missing Algorithms in Section \ref{sec::n/2}}
\begin{algorithm}[ht!]
	\caption{$\RR(i,j,S)$}
	\label{alg:rr}
	\begin{algorithmic}[1]
		\REQUIRE Agents $i,j$ and bundle $S$.
		\ENSURE A 2-partition $\mathbf{B}=(B_i,B_j)$ of $S$.
            \STATE $t \gets i$. \COMMENT{Round Robin starts with agent $i$}
            \WHILE{$S \neq \emptyset$}
                \STATE $e^*\gets \arg\max_{e\in S} v_t(e)$, tie breaks arbitrarily;
                \STATE $B_t \gets B_t \cup \{e^*\}$;
                \STATE $S \gets S \setminus \{e^*\}$;
                \STATE $t \gets \{i,j\} \setminus \{t\}$; \COMMENT{Switch turns}
            \ENDWHILE
	\end{algorithmic}
    \end{algorithm}

    \begin{algorithm}[ht!]
	\caption{ $\MAXU(i,j,S)$ }
	\label{alg:greedy}
	\begin{algorithmic}[1]
		\REQUIRE Agents $i,j$ and bundle $S$.
		\ENSURE A 2-partition $\mathbf{B}=(B_i,B_j)$ of $S$.
            \FOR{$e \in S$}
                \STATE $t \gets \arg\max_{t\in \{i,j\}}v_t(e)$. If there is a tie, set $t\gets i$;
                \STATE $B_t \gets B_t \cup \{e\}$;
            \ENDFOR
	\end{algorithmic}
    \end{algorithm}

\subsection{Missing Proofs from Section \ref{sec::n/2}}

\subsubsection{Proof of Proposition \ref{prop::rr}}

\begin{proof}[Proof of Proposition \ref{prop::rr}]
        For the first part of the statement, since agent $j$ picks first, she does not envy agent $i$, that is, $v_j(B_j)\geq v_j(B_i)$.
        For agent $i$, as $\RR$ presents an EF1 allocation and the maximum marginal value of an item is at most 1,
        we have $v_i(B_j)-v_i(B_i)\leq 1$.

        For the second part of the statement, if allocating $B_j$ (resp. $B_i$) to $j$ (resp. $i$) is not locally envy-freeable, then $v_i(B_j)+v_j(B_i) > v_i(B_i)+v_j(B_j)$, implying
        $
        v_j(B_j)-v_j(B_i)<v_i(B_j)-v_i(B_i)\leq 1.
        $
    \end{proof}

\subsubsection{Proof of Proposition \ref{prop::greedy}}
    \begin{proof}[Proof of Proposition \ref{prop::greedy}]
        By the allocation rule of $\MAXU$ and the fact that ties are broken towards agent $j$, item $e$ must be allocated to agent $j$ and hence $v_j(B_j)\geq 1$.
    \end{proof}

\subsubsection{Proof of Lemma \ref{lem::t-sum-w-bound}}
    \begin{proof}[Proof of Lemma \ref{lem::t-sum-w-bound}]
        According to Line \ref{step:mainAlg-Sub2-for-f} of Algorithm \ref{alg:multigraph-additive}, $\STWO$ starts from agents $f,g$, between whom there are two arcs in $G'$.
    Suppose $f$ envies $g$ when restricting to the partial allocation of $E_{f,g}$, then Algorithm \ref{alg:multigraph-additive} implement $\STWO$ in the out-tree with root $f$, from $f$ to the leaves.
    We will first prove that for agent $f$, the lemma statement holds, i.e., $t_f+\sum_{j\in \Ch(f)}w_j \leq w_f$, which will be used as the base case for the induction proof for Lemma \ref{lem::t-sum-w-bound}.

    Now let us prove $w_f\leq 1$. When computing $w_f$, $t_g=0$ then by Line \ref{step:sub2-defining-w} of $\STWO$, we have $w_f=(v_f(A^f_g)-v_f(A^g_f))^+$, which is at most 1 as $A^f_g,A^g_f$ are determined by $\RR$ and by Proposition \ref{prop::rr}. Therefore, $w_f\leq 1$ holds.
    We then consider possible cases for $Q_i$'s.
    Let $Q_i$'s and $Q'_1$ denote the corresponding sets at the end of $\STWO(f)$.
    
    \medskip
    \emph{Case 1: $Q_1\neq \emptyset$.} Let $\ell$ be the unique agent in $Q_1$. For this case, we will prove (i) $t_f=0$, (ii) $w_{\ell} \leq w_f$, and (iii) for any $j \in \Ch(f), j \neq \ell$, $w_j = 0$.
    It is not hard to verify that (i), (ii) and (iii) together imply $t_f+\sum_{j\in \Ch(f)}w_j \leq w_f$.

    For property (i), it suffices to prove $v_f(A^f_k)+t_k-v_f(A_f) \leq 0$. Recall $k$ is defined in Line \ref{step:sub2-defining-w} of Algorithm $\STWO$ and never gets updated when executing $\STWO(f)$. 
    First we show $v_f(A^{\ell}_f) \geq 1$.
    As arc $(f,\ell)$ exists in $G'$, there exists an item $e'\in E_{f,\ell}$ such that $v_{\ell}(e')\geq 1$.
    Due to the condition of Line \ref{step:sub2-Q1-for-conditon} in $\STWO$, we have $v_{\ell}(S^{\ell}_f) \geq v_{\ell}(S^f_{\ell})$, equivalent to $v_{\ell}(A^{\ell}_f) \geq v_{\ell}(A^f_{\ell})$ as $A^{\ell}_f=S^{\ell}_f$ and $A^f_{\ell} = S^f_{\ell}$.
    Since $v_{\ell}(e')\geq 1$, we have $v_{\ell}(A^{\ell}_f) = \max(v_{\ell}(A^{\ell}_f), v_{\ell}(A^f_{\ell})) \geq 1$.
    By $\MAXU$, we have $v_f(A^{\ell}_f) \geq v_{\ell}(A^f_{\ell})) \geq 1$. We split the remaining proof for property (i) by considering the possibilities of $k$.

    If $k=g$, we have
    \begin{align*}
    t_f &=\left(v_f(A^f_g) + t_g - v_f(A_f)\right)^+ \\
    &\leq \left( v_f(A^f_g) + t_g - v_f(A^g_f) - v_f(A^{\ell}_f) \right)^+ \\
    &\leq \left( w_f- v_f(A^{\ell}_f) \right)^+ = 0.
    \end{align*}
    The first inequality transition is due to $A^g_f\cup A^{\ell}_f \subseteq A_f$; note $\ell\in Q_1$ can never be $k$. The second inequality transition follows from the definition of $w_f$ and the last equality transition is due to $w_f\leq 1$ and $v_f(A^{\ell}_f)\geq 1$.
    
    If $k\neq g$, we have 
    \begin{align*}
        t_f=\left(v_f(A^f_k) + t_k - v_f(A_f) \right)^+ &\leq \left( v_f(A^f_k) - v_f(A^k_f) - v_f(A^{\ell}_f) \right)^+ \\
        &\leq \left( 1-v_f(A^{\ell}_f) \right)^+ = 0.
    \end{align*}
    The first inequality transition is due to $A^k_f\cup A^{\ell}_f \subseteq A_f$ and the fact that $t_k=0$ at the moment when computing $t_f$; note at that moment, for every $j$ the direct successor of $i$, $t_j=0$ holds. The second inequality arises from the fact that $\ai{f}{\ell}, \ao{f}{\ell}$ are determined by $\RR$ and from Property \ref{prop::rr}.

    For property (ii), since $\ell \in Q_1$, by the condition of Line \ref{step:sub2-Q1-for-conditon} in $\STWO$, we have $v_{\ell}(A^{\ell}_f)-v_{\ell}(A^{f}_{\ell}) \leq w_f$; recall that $S^{\ell}_f=A^{\ell}_f$ and $S^{f}_{\ell}=A^{f}_{\ell}$.
    Then by $t_f=0$, we have $w_i \geq v_{\ell}(A^{\ell}_f) +t_f -v_{\ell}(A^{f}_{\ell}) \geq w_{\ell}$ where the second inequality transition follows from the definition of $w_{\ell}$.
    Last for property (iii), for $j\in Q'_1\cup Q_2 \cup Q_5$, $A^f_j=T^f_j$ is determined by $\RR(j,f,E_{f,j})$ where $j$ picks first. Thus for these $j$, we have $v_j(A^f_j) \geq v_j(A^j_f)$, together $t_f=0$ implying $w_j=0$.
    For $j\in Q_3$, $A^f_j=S^f_j$ and $j$ does not satisfy the condition of Lines \ref{step:sub2-Q1-for-conditon} or \ref{step:sub2-Q2-for-conditon} in $\STWO$, and again, $v_j(A^f_j) \geq v_j(A^j_f)$ (and hence $w_j=0$).
    Next, we claim that $Q_4=\emptyset$ in Case 1. We now consider the moment when executing Line \ref{step:sub2-defining-S}. For any $j\notin R$, we have
    \begin{align*}
        v_f(A_f)+v_j(A^f_j) &\geq v_f(A^{\ell}_f)+v_f(A^j_f)+v_j(A^f_j) \\
        &\geq 1+ v_f(A^f_j)-1+v_j(A^j_f) = v_f(A^f_j)+v_j(A^j_f),
    \end{align*}
    where the first inequality transition follows from $A^j_f\cup A^{\ell}_f\subseteq A_f$ and the second inequality transition is due to $v_f(A^{\ell}_f)\geq 1$; $v_f(A^j_f)-v_f(A^f_j)\leq 1$ by Proposition \ref{prop::rr};
    $v_j(A^f_j) \geq v_j(A^j_f)$ as $A^f_j=T^f_j$ at that moment.
    Therefore, $S=\emptyset, Q_4=\emptyset$, and therefore, property (iii) holds. Up to now, we show that if Case 1 occurs, then $t_f+\sum_{j\in \Ch(f)}w_j \leq w_f$.

    \medskip
    \emph{Case 2: $Q_1=\emptyset$ and $Q_2\neq \emptyset$}. For this case, $Q'_1=\emptyset$. Fix $\ell \in Q_2$ and we will prove (i) $t_f=0$, and (ii) for any $j \in \Ch(f), w_j \leq 0$, together implying $t_f+\sum_{j\in \Ch(f)}w_j \leq w_f$.
    
    For property (i), by the condition of Line \ref{step:sub2-Q2-for-conditon}, we have $w_f<v_{\ell}(S^{\ell}_f)-v_{\ell}(S^f_{\ell})$. 
    Since arc $(f,\ell)$ exists in $G'$, there exists $e'\in E_{f,\ell}$ such that $v_{\ell}(e')=1$.
    By $\MAXU$, $v_{\ell}(S^f_{\ell})\geq 1$, and thus, $v_{\ell}(S^{\ell}_f) >w_f+1$. Moreover, the value of $S^{\ell}_f$ for agent $f$ is higher than that of agent $j$, i.e., $v_f(S^{\ell}_f)>v_{\ell}(S^{\ell}_f)$, and hence $v_f(S^{\ell}_f)>w_f+1$. Note that $A^{\ell}_f=T^{\ell}_f$ and $A^f_{\ell}=T^f_{\ell}$ and bundles $T^{\ell}_f,T^f_{\ell}$ are determined by $\RR({\ell},f,E_{f,{\ell}})$.
    Since $e'\notin S^{\ell}_f$, one can verify that (agent $f$ picks first) $\RR(f,{\ell}, E_{f,{\ell}}\setminus\{e'\}) = (T^{\ell}_f, T^f_{\ell}\setminus\{e'\} )$\footnote{We can let agents pick the same set of items as that chosen in $\RR(\ell,f,E_{f,\ell})$ when there are multiple items having the same value.}, 
    and moreover, $v_f(T^{\ell}_f)\geq \frac{v_f(S^{\ell}_f)}{2}$ as $S^{\ell}_f\subseteq E_{f,\ell}\setminus\{e'\} $, together $T^{\ell}_f=A^{\ell}_f$ implying $v_f(A^{\ell}_f)\geq \frac{w_f+1}{2}$. Given $w_f \leq 1$, we have $v_f(A^{\ell}_f) \geq w_f$. In the following, we show $t_f=0$.

    If $k=g$, we have
    \begin{align*}
    t_f&=\left(v_f(A^f_g) + t_g - v_f(A_f)\right)^+ \\
    &\leq \left( v_f(A^f_g) + t_g - v_f(A^g_f) - v_f(A^{\ell}_f) \right)^+ 
    \leq \left( w_f- v_f(A^{\ell}_f) \right)^+ = 0.
    \end{align*}
    The first inequality transition is due to $A^g_f\cup A^{\ell}_f \subseteq A_f$; note that $j\in Q_2$ cannot be $k$. The second inequality transition follows from the definition of $w_f$ and the last equality transition is due to $v_f(A^{\ell}_f) \geq w_f$.

    If $k\neq g$, it holds that $v_f(A^{\ell}_f)+v_f(A^g_f)\geq w_f+v_f(A^g_f)\geq 1$ due to Proposition \ref{prop::bound}.
    Then we have
    \begin{align*}
        t_f&=\left(v_f(A^f_k) + t_k - v_f(A_f) \right)^+ \\
        &\leq \left( v_f(A^f_k) - v_f(A^k_f) - v_f(A^{\ell}_f) - v_f(A^g_f) \right)^+ 
    = 0.
    \end{align*}
    The first inequality transition is due to $A^k_f\cup A^{\ell}_f\cup A^g_f \subseteq A_f$ and the fact that $t_k=0$ at the moment when computing $t_f$.
    For the second inequality transition, by Proposition \ref{prop::rr}, $v_f(A^f_k)-v_f(A^k_f)\leq 1$ and $v_f(A^{\ell}_f)+v_f(A^g_f)\geq 1$ hold.
    Therefore, $t_f=0$ in Case 2.

    For property (ii), similar to the argument for \emph{Case 1}, one can verify that for $j\in Q_2\cup Q_3 \cup Q_5$, we have $A^f_j=T^f_j, A^j_f=T^j_f$ and thus $v_j(A^f_j) \geq v_j(A^j_f)$, together $w_f=0$ implying $w_j=0$.
    Last we prove $Q_4=\emptyset$ in Case 2.
    We now consider the moment when executing Line \ref{step:sub2-defining-S}. For any $j \in \Ch(f)\setminus R$, we have $j\neq \ell,g$ and
    $$
    \begin{aligned}
    v_f(A_f)+v_j(A^f_j)&\geq v_f(A^{\ell}_f) + v_f(A^g_f)+v_f(A^j_f)+v_j(A^f_j) \\
    & \geq 1+ v_f(A^f_j)-1+v_j(A^j_f) \\
    & = v_f(A^f_j)+v_j(A^j_f),
    \end{aligned}
    $$
    where the first inequality transition follows from $A^j_f\cup A^{\ell}_f\cup A^g_f\subseteq A_f$ and the second inequality transition is due to $v_f(A^{\ell}_f)+v_f(A^g_f)\geq 1$; $v_f(A^j_f)-v_f(A^f_j)\leq 1$ by Proposition \ref{prop::rr};
    $v_j(A^f_j) \geq v_j(A^j_f)$ as $A^f_j=T^f_j$ at that moment.
    Therefore, $S=\emptyset, Q_4=\emptyset$, and therefore, property (ii) is valid. Up to now, we show that if Case 2 occurs, then $t_f+\sum_{j\in \Ch(f)}w_j \leq w_f$.


    \medskip
    \emph{Case 3: $Q_1\cup Q_2=\emptyset$ and $t_f=0$}. Again $Q_1=\emptyset$ implies $Q'_1=\emptyset$. 
    For any $j\in Q_3\cup Q_5$, we have $v_j(A^f_j) \geq v_j(A^j_f)$, together $t_f=0$ implying $w_j=0$. Then if $Q_4=\emptyset$, we have $t_f+\sum_{j\in \Ch(f)} w_j = 0 \leq w_f$.

    For the case where $Q_4\neq \emptyset$, let $\ell$ be the unique agent in $Q_4$. 
    Suppose that agent $f$ receives bundle $A'_f$ right before Line \ref{step:sub2-defining-S} 
    of $\STWO$. Then $A'_f=A_f \cup T^{\ell}_{f} \setminus T^{f}_{\ell} $. By the condition of Line \ref{step:sub2-defining-S}, we have
    $$
    v_f(A'_f) + v_{\ell}(T^f_{\ell}) < v_f(T^f_{\ell}) + v_{\ell}(T^{\ell}_f).
    $$
    Since $\ell$ and $f$ swap their items from $E_{f,\ell}$, we have $A^f_{\ell}=T^{\ell}_f$ and $A^{\ell}_f=T^f_{\ell}$. From the above inequality, we derive
    $$
    \begin{aligned}
    v_{\ell}(A^{\ell}_f) - v_{\ell}(A^f_{\ell}) & < v_{f}(A^{\ell}_f) - \left( v_f(A_f) + v_f(A^f_{\ell}) - v_f(A^{\ell}_f) \right) \\
    & \leq   v_f(A_f^{\ell}) - v_f(A^f_{\ell}) -  v_f(A^g_f) \\
    & \leq 1 - v_f(A^g_f)\leq w_f,
    \end{aligned}
    $$
    where the second inequality transition follows from $A^{\ell}_f\cup A^g_f \subseteq A_f$; the third inequality transition is due to Proposition \ref{prop::rr};
    the last inequality transition follows from Proposition \ref{prop::bound}.
    Then as $t_f=0$, it holds that $w_{\ell} = (v_{\ell}(A^{\ell}_f) - v_{\ell}(A^f_{\ell}) )^+ \leq w_f$. Thus, $t_f+\sum_{j\in \Ch(f)} w_j \leq w_f$ holds for Case 3.

    \medskip
    \emph{Case 4:} $Q_1\cup Q_2=\emptyset$ and $t_f>0$. Also, $Q'_1=\emptyset$. For this case, $t_f=t_k+v_f(A^f_k)-v_f(A_f)$ where $k$ is determined in Line \ref{step:sub2-defining-k} of $\STWO(f)$ and $t_k$ is the one right after the termination of $\STWO(f)$; note that $t_k$ can be updated in later subroutines and if $k\neq g$, $t_k=0$.
    Next, we present upper bounds of $w_j$ based on the classification of $j\in \Ch(f)$ and remark that $w_j$ is computed in $\STWO(j)$, rather than $\STWO(f)$.

    For $j\in Q_3$, we will show $w_j\leq 0$.
    Since $A^j_f=S^j_f$ and $A^f_j=S^f_j$, we have $v_f(A^j_f) \geq v_j(A^j_f)$ and $v_j(A^f_j) \geq v_f(A^f_j)$. Moreover $v_j(A^f_j) \geq 1$ as there exists $e'\in S^f_j$ such that $v_j(e')=1$; which is because arc $(f,j)$ exists in $G'$.
    Then by substituting $t_f$,  we have
    $$
    \begin{aligned}
        w_j& =\left( t_k+v_f(A^f_k)-v_f(A_f)+v_j(A^j_f)-v_j(A^f_j) \right)^+ \\
        & \leq \left( t_k+v_f(A^f_k)-v_f(A_f)+v_f(A^j_f)-v_j(A^f_j) \right)^+.
    \end{aligned}
    $$
    Then if $k=g$, since $A^j_f\cup A^g_j \subseteq A_f$, we have
    $$
    w_j \leq \left( t_g+v_f(A^f_g)-v_f(A^g_f) -v_j(A^f_j) \right)^+ \leq \left( w_f -1 \right)^+ = 0.
    $$
    If $k\neq g$, then $t_k=0$ when computing $w_j$ and thus we have
    \begin{align*}
        w_j &\leq \left( v_f(A^f_k)-v_f(A_f)+v_f(A^j_f)-v_j(A^f_j) \right)^+ \\
        & \leq \left( v_f(A^f_k)-v_f(A^k_f) - 1 \right)^+ = 0,
    \end{align*}
    where the second inequality transition is due to $A^k_f\cup A^j_f \subseteq A_f$; note $j\neq k$.

    For $j\in Q_4$, recall, from the proof for Case 3, that $$
     v_{j}(A^{j}_f) - v_{j}(A^f_{j})  < v_{f}(A^{j}_f) - \left( v_f(A_f) + v_f(A^f_{j}) - v_f(A^{j}_f) \right),
     $$
    Then we have
    $$
    \begin{aligned}
    &w_j=\left( t_k+v_f(A^f_k)-v_f(A_f)+v_j(A^j_f)-v_j(A^f_j) \right)^+ \\
    & \leq \left(  t_k+v_f(A^f_k) - v_f(A_f) + v_{f}(A^{j}_f) - \left( v_f(A_f) + v_f(A^f_{j}) - v_f(A^{j}_f) \right) \right)^+\\
    & \leq \left(  t_k+v_f(A^f_k) - v_f(A^f_j)\right)^+,
    \end{aligned}
    $$
    deriving the desired upper bound.

    For $j\in Q_5$, we have $A^j_f=T^j_f$ and $A^f_j=T^f_j$. Suppose that agent $f$ receives $A'_f$ right before Line \ref{step:sub2-defining-S} of $\STWO(f)$.
    At that moment, if $j$ does not meets the condition in Line \ref{step:sub2-defining-S}, we have $v_f(A^f_j)+v_j(A^j_f)\leq v_f(A'_f)+v_j(A^f_j)$, equivalent to $v_f(A^f_j)-v_f(A'_f) \leq v_j(A^f_j) - v_j(A^j_f)$.
    It is not hard to verify that $v_f(A_f)\geq v_f(A'_f)$ due to the definition of envy-freeability (if swapping happens).
    Accordingly, we have $v_f(A^f_j)-v_f(A_f) \leq v_j(A^f_j) - v_j(A^j_f)$.
    On the other hand, if when executing Line \ref{step:sub2-defining-S}, $j$ does meets the condition in Line \ref{step:sub2-defining-S}, 
    $v_f(A^f_j)-v_f(A_f) \leq v_j(A^f_j) - v_j(A^j_f)$ still holds: since there exists an agent $\ell \in Q_4$ such that $v_f(T_\ell^f) \geq v_f(T_j^f) = v_f(A_j^f)$ and $T_\ell^f \subseteq A_f$, the left hand side is negative while the right hand side is non-negative.
    Thus, we have
    $$
    \begin{aligned}
    w_j&=\left( t_k+v_f(A^f_k)-v_f(A_f)+v_j(A^j_f)-v_j(A^f_j) \right)^+  \\
    & \leq \left( t_k+v_f(A^f_k)-v_f(A_f) + v_f(A_f)-v_f(A^f_j)\right)^+ \\
    & = \left( t_k+v_f(A^f_k)-v_f(A^f_j)\right)^+,
    \end{aligned}
    $$
    deriving the desired bound.

    For $j=k$, we can have $j\neq g$ as $g\notin \Ch(f)$ holds. Also $t_j=0$ at the moment when computing $t_f$ in $\STWO(f)$. Moreover, $j\notin Q_4$ as otherwise, $t_f=0$, contradicting the condition of Case 4. 
    Then it must hold that $j\in Q_5$ and hence $v_f(A_f)+v_j(A^f_j) \geq v_f(A^f_j)+v_j(A^j_f)$, implying
    $$
    \begin{aligned}
    w_j&=\left( t_j+v_f(A^f_j)-v_f(A_f)+v_j(A^j_f)-v_j(A^f_j) \right)^+=0.
    \end{aligned}
    $$
    Up to here, we present the upper bound of $w_j$ for each possibility of $j\in \Ch(f)$.

    Next we proceed to bound $t_f+\sum_{j\in \Ch(f)}w_j$ for Case 4. For simple notations, let $B=\{j \in Q_4\cup Q_5 \setminus\{k\} \mid w_j >0 \}$. Then we have
    \begin{align*}
         t_f+\sum_{j\in \Ch(f)}w_j & = t_k + \ao{f}{k} - \ax{f} + \sum_{j\in B}\left(t_k + \ao{f}{k} - \ao{f}{j} \right)^+\\
        & = t_k + \ao{f}{k} - \ax{f} + \sum_{j\in B} \ai{f}{j}  + \sum_{j\in B}\left(t_k + \ao{f}{k} - \ao{f}{j} - \ai{f}{j}\right)\\
        & = t_k + \ao{f}{k} - \ax{f} + \sum_{j\in B} \ai{f}{j}  + \sum_{j\in B}\left(t_k + \ao{f}{k} - v_i(\E{f}{j})\right)\\
        & \leq t_k + \ao{f}{k} - \ax{f} + \sum_{j\in B} \ai{f}{j},
    \end{align*}
    where the inequality transition is because each agent in $Q_4 \cup Q_5$ does not satisfy the condition in Line \ref{step:sub2-defining-k}. 
    Accordingly, if $k=g$, we have
    $$
    \begin{aligned}
    t_f+\sum_{j\in \Ch(f)}w_j & \leq t_g + \ao{f}{g} - \ax{f} + \sum_{j\in B} \ai{f}{j} \\
    & \leq t_g+v_f(A^f_g) - v_f(A^g_f) \leq w_f.
    \end{aligned}
    $$
    If $k\neq g$, we have
     $$
    \begin{aligned}
    t_f+\sum_{j\in \Ch(f)}w_j & \leq t_k + \ao{f}{k} - \ax{f} + \sum_{j\in B} \ai{f}{j} \\
    & \leq t_k+v_f(A^f_k) - v_f(A^k_f) - v_f(A^g_f) \\
    & \leq 1-v_f(A^g_f) \leq w_f,
    \end{aligned}
    $$
    where the last inequality transition is due to Propositions \ref{prop::rr} and \ref{prop::bound}.
    Up to here, we complete the base case proof for Lemma \ref{lem::t-sum-w-bound}.

    Then we perform induction in the breadth-first search fashion. Suppose for every $p$ at depth $d$ of the rooted tree starting from $f$, we have $t_p+\sum_{j\in \Ch(p)}w_j \leq w_p$. Since $w_f\leq 1$, we have $w_p\leq 1$, implying $w_j\leq 1$ for all $j\in \Ch(p)$.
    Then using the same argument as that for $f$, we can prove that for every $p$ at depth $d+1$ of the rooted tree starting from $f$, we also have $t_p+\sum_{j\in \Ch(p)}w_j \leq w_p$, completing the proof of Lemma \ref{lem::t-sum-w-bound}.
    \end{proof}

\subsubsection{Proof of Lemma \ref{lem::efable-phase1}}
    \begin{proof}[Proof of Lemma \ref{lem::efable-phase1}]
    In this proof, for any $i,j$, we refer $A_i$ and $A^j_i$ to the corresponding bundle at the moment right before Phase 2 (Line \ref{step:mainAlg-start-Phase2}). 
    Let us focus on some vertex $i$ in component $C$ of $G'$, and let $s$ be $i$'s direct predecessor and $j$ be $i$'s direct successor if any. We will prove that $t_i+v_i(A_i)\geq t_p+v_i(A^i_p)$ for $p=s,j$. 
    The proof is split by discussing the possibility of the component $C$. 
    \medskip
    
    \emph{Case 1:} $C$ has a cycle of length 3 or more. In this case, $t_i,t_s,t_j = 0$. First consider $E_{i,s}$ allocated by $\RR$ with $i$ picking first, and thus, $v_i(A^s_i) \geq v_i(A^i_s)$, with $t_i=t_s=0$ implying the desired inequality. Moreover, as arc $(s,i)$ exists in $G'$, there exists $e'\in E_{i,s}$ such that $v_i(e')=1$, and hence, $v_i(A^s_i)\geq 1$.
    For $E_{i,j}$, Proposition \ref{prop::rr} gives that $\dif{i}{j}\leq 1$, implying
                $$
                    \ax{i} + t_i - \ao{i}{j} - t_j \geq \ai{i}{j} +\ai{i}{s} - \ao{i}{j} \geq 0.
                $$
    \medskip
    
    \emph{Case 2: $C$ has 2-cycle and $i$ is handled by $\SONE$.} Similarly, $t_i,t_j = 0$ as the allocation for $i,j$ and their successors are determined by $\SONE$. Then by arguments similar to that in Case 1, one can derive the desired inequality for $i,j$.
    As for $i,s$, if the allocation for $s$ and their successor is determined by $\SONE$, then one can derive the desired inequality for $i,s$ by arguments similar to that of Case 1.
    
    We then prove for the case where the allocation for $s$ and their successor is determined by $\STWO$, and hence, $s$ and $i$ forms the unique cycle with length 2. Note that $t_s\leq w_s=(t_i+v_s(A^s_i)-v_s(A^i_s))^+$ and the algorithm ensures that $\ndif{i}{s} \geq \dif{s}{i} $ (i.e., local envy-freeability), we have 
          $$
                    \ax{i} - \ao{i}{s} - t_s \geq \ai{i}{s} - \ao{i}{s} + (\dif{i}{s})^+ = 0,
     $$
     completing the proof for Case 2.

     \medskip

     \emph{Case 3: $C$ has a 2-cycle and $i$ is handled by $\STWO$.}
     In the proof of this case, let $Q_i$'s and $k$ be those defined when executing $\STWO(i)$ and let $t'_p$'s be the $t_p$ at the beginning of $\STWO(i)$.
     First consider $i,s$. If $s=k$ or {agent $k$ receives $T^i_k$}, then $t_i \geq t'_k + v_i(A^i_k) - v_i(A_i) \geq t_s + v_i(A^i_s) - v_i(A_i)$ holds; 
     Thus $\ax{i} + t_i \geq \ao{i}{s}+ t_s$.
     {If $s\neq k$ and agent $k$ receives $T^k_i$}, we have $T^i_k\subseteq A_i$, and accordingly,
     $v_i(A_i)+t_i\geq v_i(T^i_k) + t'_k > v_i(A^i_s) + t_s $, the desired inequality.
     
     Next for $j$, we consider several situations. 

     If $j\in Q_1\cup Q'_1$, we have $Q_1\neq \emptyset$ in this case and thus according to the proof of Case 1 of Lemma \ref{lem::t-sum-w-bound}, $t_i=0$ and hence $t_j=0$. 
     When executing $\STWO(i)$, 
     $j$ meets the condition of Line \ref{step:sub2-checking-rr-k}, i.e., $v_i(E_{i,j})\leq v_i(T^i_k)+t'_k$. 
     If $k\notin Q_4$, {then $A^i_k=T^i_k$; note $T^i_s$ is defined to be $A^i_s$ in Line \ref{step:sub2-defining-k}.} We have
     $$
        v_i(A_i)+t_i = v_i(A_i)+\left( v_i(T^i_k)+t'_k-v_i(A_i)\right)^+ \geq v_i(E_{i,j}) \geq v_i(A^i_j)+t_j.
     $$
     If $k\in Q_4$ (and hence $k\neq s$), then {$T^i_k\subseteq A_i$} based on Line \ref{step:sub2-Q4-swap}. Hence, we have
     $$
        v_i(A_i)+t_i \geq v_i(T^i_k) + t'_k \geq v_i(E_{i,j}) \geq v_i(A^i_j)+t_j,
     $$
     where the first inequality is due to $t'_k=0$ and $T^i_k\subseteq A_i$.

     If $j\in Q_3$, 
     similarly when executing $\STWO(i)$, 
     $j$ meets the condition of Line \ref{step:sub2-checking-rr-k}, i.e., $v_i(E_{i,j})\leq v_i(T^i_k)+t'_k$. 
     If $t_j=0$, then the proof is identical to the proof for the case of $j\in Q_1\cup Q'_1$ and we omit it.
     For $t_j>0$, we have
     $w_j>0$ and
     $$
     \begin{aligned}
       & v_i(A_i)+t_i -v_i(A^i_j) - t_j \\
       & \geq v_i(A_i)+t_i- v_i(A^i_j) - t_i -v_j(A^j_i) + v_j(A^i_j) \\ 
       &\geq v_i(A_i)-v_i(A^i_j)-v_i(A^j_i)+v_j(A_j^i) \geq 0,
     \end{aligned}
     $$
     where the last two inequality transitions follows from that $E_{i,j}$ is allocated by $\MAXU$, i.e.,  $v_j(A^i_j) \geq v_i(A^i_j)$ and $v_i(A^j_i)\geq v_j(A^j_i)$.

     If $j\in Q_2$, then $E_{i,j}$ is allocated by $\RR$ and $j$ picks first. As we have proved in Case 2 of the proof of Lemma \ref{lem::t-sum-w-bound}, $Q_2\neq \emptyset$ implies $w_i=0$ and hence $t_j=0$ by Lemma \ref{lem::t-sum-w-bound}.
     {Since $j\in Q_2$ implies that $j$ meets the condition of Line \ref{step:sub2-checking-rr-k}, i.e., $v_i(E_{i,j})\leq v_i(T^i_k)+t'_k$, and therefore, the remaining proof of this case is identical that the proof of case $j\in Q_1\cup Q'_1$ with $t_j=0$. We omit it.}

     If $j\in Q_4$, by Line \ref{step:sub2-Q4-swap}, we have $v_i(A^j_i)\geq v_i(A^i_j)$ (note $A^j_i, A^i_j$ are bundles after swap). If $t_j=0$, we have
     $
        v_i(A_i)+t_i\geq  v_i(A^j_i)\geq v_i(A^i_j) +t_j
     $.
     If $t_j>0$, then $w_j>0$ by Proposition \ref{prop::bound} we have $t_j\leq w_j=t_i+v_j(A^j_i) -v_j(A^i_j) $. Then the following holds,
     $$
        v_i(A_i)+t_i-v_i(A^i_j)-t_j \geq  v_i(A_i)+t_i - v_i(A^i_j)- t_i -v_j(A^j_i) +v_j(A^i_j) \geq 0,
     $$
        where the last inequality transition is guaranteed by the algorithm after swapping in Line \ref{step:sub2-Q4-swap}.

        If $j\in Q_5$, we first consider the case of $t_j=0$. It is not hard to verify that no matter whether $A^i_k=T^i_k$ or not, we have
        $v_i(A_i)+t_i\geq v_i(T^i_k) + t'_k \geq v_i(T^i_j) = v_i(A^i_j) + t_j$.
        {If $t_j>0$ (and hence $w_j>0$), the proof is identical to the proof of the case $j\in Q_4$ with $t_j>0$. We omit it.}

        Up to here, we complete the proof.
    \end{proof}

\subsubsection{Proof of Lemma \ref{lem::final-efable}}

\begin{proof}[Proof of Lemma \ref{lem::final-efable}]
        According to Phase 2 of Algorithm \ref{alg:multigraph-additive}, $p_j\leq t_j$ for all $j\in N$.
        We first prove the envy-freeness between agents $i,j$ not adjacent in $G'$. 
        Since the partial allocation of $E_{i,j}$ is locally envy-freeable, there must be an agent not envying the other when restricting to the allocation of $E_{i,j}$.
        Without loss of generality, let agent $j$ be such an agent.
        
        {It is worthwhile to mention that the payment of an agent can be decreased more than once. We first prove that if the payment for agent $j$ decreases when executing Phase 2 for $E_{i,j}$, then the sum of agent $j$'s value and the payment does not get decreased.}
        Moreover, let $p'_j$ and $p''_j$ be respectively the payment for $j$ right before and right after decreasing the payment for $j$. Let $p'_i$ be the payment for agent $i$ when executing Phase 2 for $E_{i,j}$. Then our task is to prove $v_j(A^i_j)\geq p'_j - p''_j$ and we will prove it by induction. Let $B_j$ and $B_i$ be respectively the bundle of agent $j$ and $i$ right after Line \ref{step:mainAlg-last-adjust-p}.

        For the base case where $p'_j=t_j, p'_i=t_i$, we have
        $$
        \begin{aligned}
        p'_j-p''_j &\leq p'_j-v_i(B_i)-p'_i+v_i(A^i_j) \\
        & \leq 1-v_i(B_i)+v_i(A^j_i)-p'_i+v_i(A^i_j)-v_i(A^j_i) \\
        &= 1- (v_i(B_i)-v_i(A^j_i) + p'_i) + v_i(A^i_j)-v_i(A^j_i)\\
        & \leq v_j(A^i_j)-v_j(A^j_i)\leq v_j(A^i_j).
        \end{aligned}
        $$
        The second inequality is due to $p'_j \leq t_j \leq 1$. The third inequality transition follows from (i) before receiving $A^j_i$, the sum of agent $i$'s value and payment is at least 1; note for based case, this is true as $v_i(B_i\setminus A^j_i)+t_i\geq 1$; for other cases, this is also true due to the induction assumption;
        and (ii) $v_i(A^j_i)+v_j(A^i_j)\geq v_i(A^i_j)+v_j(A^j_i)$ ensured by local envy-freeability when restricting to $E_{i,j}$.
        Then one can make the assumption that for each of the first $p$ time when the payment of some agent is decreased, the sum of her value and the payment does not get decreased.
        Then the argument for the $p+1$ time is similar to the argument for the base case and we omit it. Up to here, we show that after decreasing the payment for agent $j$, the sum of her value and payment is at least the sum of value and payment before allocating $\E{i}{j}$.
        
        Next we prove that after allocating $E_{i,j}$ and adjusting the payment (if applicable), for agents $i,j$, no one envies the other.
        If $p'_j$ does not get decreased in this round, then according to Line \ref{step:mainAlg-last-if}, we have $v_i(A_i)+p'_i\geq v_i(A^i_j)+p'_j$ and hence agent $i$ does not envy $j$ in Line \ref{step:mainAlg-last-adjust-p}.
        For agent $j$, by the above argument on the monotonicity of the sum of value and payment,
        it is not hard to verify that $v_j(B_j\setminus A^i_j) + p'_j \geq 1$ due to Lemma \ref{lem::efable-phase1} and $v_j(e)=1$ for some $e\in M$.
        Accordingly, we have
        $$
        \begin{aligned}
            &v_j(B_j) + p''_j- v_j(A^j_i) - p'_i \\
            & \geq v_j(B_j)+p''_j-v_j(A^i_j) +  v_j(A^i_j) - v_j(A^j_i) - p'_i \\
            &\geq v_j(B_j) + p''_j - \ai{j}{i} - p'_i\\
            &\geq v_j(B_j\setminus A^i_j) + p''_j - 1 \geq 0,
        \end{aligned}
        $$
        where the second inequality is due to $p_i\leq t_i\leq 1$ and $v_j(A^i_j)\geq v_j(A^j_i)$ and the third inequality transition is due to the monotonicity that we just proved. Hence, agent $j$ does not envy $i$ in Line \ref{step:mainAlg-last-adjust-p}.
        Since for any agent, the sum of value and payment never decreases and the payment never increase in the later subroutine, we can claim that agent $i$ and $j$ do not envy each other in $\{\mathbf{A},\mathbf{p}\}$.

        If $p'_j > p'' _j$, according to Line \ref{step:mainAlg-last-adjust-p}, we have $p''_j=(v_i(B_i)+p'_i-v_i(A^i_j))^+$, implying
        $v_i(B_i) + p'_i \geq v_i(A^i_j) + p''_j$.
        Thus, agent $i$ is envy-free in Line \ref{step:mainAlg-last-adjust-p}. 
        As for agent $j$, again by Line \ref{step:mainAlg-last-adjust-p}, we have $p''_j\geq v_i(B_i)+p'_i-v_i(A^i_j) \geq v_i(A^j_i)+p'_i-v_i(A^i_j)$ since $A^j_i\subseteq B_i$.
       When restricting to $E_{i,j}$, the partial allocation is locally envy-freeable and hence $v_i(A^j_i) + v_j(A^i_j)\geq v_i(A^i_j) + v_j(A^j_i)$, equivalent to $v_i(A^j_i) - v_i(A^i_j) \geq v_j(A^j_i) - v_j(A^i_j)$. Accordingly, we have
       $$
        p''_j\geq v_j(A^j_i) - v_j(A^i_j) + p'_i \implies 
         v_j(A^i_j) + p'_j\geq v_j(A^j_i) + p'_i.
       $$
       Therefore, agent $j$ is also envy-free in Line \ref{step:mainAlg-last-adjust-p}.
       Since for any agent, the sum of value and payment never decreases and the payment never increase in the later subroutine, we can claim that agent $i$ and $j$ do not envy each other in $\{\mathbf{A},\mathbf{p}\}$.


       Last for any $i,j$ adjacent in $G'$, by Lemma \ref{lem::efable-phase1} and the facts that (i) the sum of value and payment for an agent never decrease, and (ii) the payment for an agent never increases,
       we can claim that agents $i,j$ do not envy each other with respect to $\{\mathbf{A}, \mathbf{p}\}$ or at the termination of Algorithm \ref{alg:multigraph-additive}. We complete the proof of the statement.    
    \end{proof}

\end{document}